\documentclass[12pt]{article}

\usepackage[utf8]{inputenc}
\usepackage[english]{babel}
\usepackage{latexsym,amsmath,amssymb,color}
\usepackage{amsmath}
\usepackage{amsfonts}
\usepackage{amsthm}
\usepackage{amssymb}
\usepackage{stmaryrd}
\usepackage{url}
\usepackage{bbold}
\usepackage[T1]{fontenc}
\usepackage{lmodern}
\usepackage[pdftex]{graphicx}
\usepackage{float}
\usepackage[round]{natbib}
\usepackage{booktabs}
\usepackage[margin=1in]{geometry}

\newtheorem{theorem}{Theorem}[section]
\newtheorem{corollary}{Corollary}[theorem]

\newtheorem{rem}[theorem]{Remark}
\newtheorem{definition}{Definition}[section]

\newcommand{\RomanNumeralCaps}[1]{\MakeUppercase{\romannumeral #1}}

\newcommand{\comment}[1]{}

\begin{document}

\begin{center}
\section*{Multivariate Hawkes-based Models in LOB: European, Spread and Basket Option Pricing}
{\sc Qi Guo}\\Department of Mathematics and Statistics, University of Calgary, 2500 University Drive NW, Calgary, AB T2N 1N4, Canada \\ 
\vspace{0.3cm}
{\sc Anatoliy Swishchuk}\\Department of Mathematics and Statistics, University of Calgary, 2500 University Drive NW, Calgary, AB T2N 1N4, Canada\\ 
\vspace{0.3cm}
{\sc Bruno R{\'e}millard}\\Department of Decision Sciences, HEC Montr\'{e}al, 3000 Chemin de la Cote-Sainte-Catherine, Montr\'{e}al, QC  H3T 2A7,
Canada\\

\end{center}

\hspace{1cm}

{\bf Abstract:} In this paper, we consider pricing of European options and spread options for Hawkes-based model for the limit order book. We introduce multivariate Hawkes process and the multivariable general compound Hawkes process. Exponential multivariate general compound Hawkes processes and limit theorems for them, namely, LLN and FCLT, are considered then. We also consider a special case of one-dimensional EMGCHP and its limit theorems. Option pricing with $1D$ EGCHP in LOB, hedging strategies, and numerical example are presented. We also introduce greeks calculations for those models. Margrabe's spread options valuations with Hawkes-based models for two assets and numerical example are presented. Also, Margrabe's spread option pricing with two $2D$ EMGCHP and numerical example are included. Basket options valuations with numerical example are included. We finally discuss the implied volatility and implied order flow. It reveals the relationship between stock volatility and the order flow in the limit order book system. In this way, the Hawkes-based model can provide more market forecast information than the classical Black-Scholes model.

\vspace{0.5cm}

{\bf Keywords}: Multivariate general compound Hawkes process (MGCHP); exponential MGCHP (EMGCHP); LLN and FCLT; limit order book; option pricing; Margrabe's spread option pricing.

\section{Introduction}

Pricing options for the limit order book  (LOB) was initiated very recently. For example, paper \cite{SR} prices European options in a discrete time model for the LOB, and hence builds a discrete time model for the structure of the limit order book, so that the price per share depends on the size of the transaction. 

Some authors used the idea to apply a binomial model for the LOB. In \cite{GS1}, they study a binomial version of the LOB model from \cite{CJP}. The dynamic programming was used to solve the problem of super-hedging
for derivatives depending on the spot price, where the marked to market value for
the value of the portfolio was taken and a linear supply curve was supposed. The limit in
continuous time for a similar model in \cite{GS2} was studied as well. 
In \cite{AG}, the authors study the problem of hedging on a model based on \cite{OW}, where the cost of illiquidity is given by a linear density. They also use the marked to
market value to define the portfolio. In \cite{DS}, 
the limit in continuous time of a binomial model for the LOB was studied. They also took
the marked to market value for the portfolio and derivatives as functions of the spot price.
Applications on a binomial model for illiquidity, was considered in \cite{SR}, where they  maximized utility of wealth using primal/dual optimization when illiquidity
is modelled with a non-linear supply curve. For orders
modelled by using Poisson processes, we refer to \cite{BL} or
\cite{FKK} and references therein.

In this paper, we consider pricing of European options and spread options for Hawkes-based model for the limit order book. Multivariate general compound Hawkes processes and their applications in limit order books were considered in \cite{Guo2020}. Multivariate General Compound Point Processes in Limit Order Books were considered in \cite{GRS}. The paper \cite{SR} studied pricing European options in a discrete time model for the limit order book. One dimensional exponential general compound Hawkes process was first introduced in \cite{S}. A level-1 limit order book with time
dependent arrival rates was considered in \cite{CERS}. Exponential multivariate general compound Hawkes processes and limit theorems for them, such as LLN and FCLT, were introduced in \cite{G}.\\
The rest of the paper is organized as follows. Section 2 introduces multivariate Hawkes process and the multivariable general compound Hawkes process. The exponential multivariate general compound Hawkes processes and limit theorems for them, namely, LLN and FCLT, are considered in Section 3. Here we also consider a special case of one-dimensional EMGCHP and its limit theorems. Option pricing with $1D$ EGCHP in LOB and numerical example are presented in Section 4. Here, we also present greeks calculations. Margrabe's spread options valuations with Hawkes-based models for two assets and numerical example are presented in Section 5. Margrabe's spread option pricing with two $2D$ EMGCHP and numerical example are presented in Section 6. Section 7 concludes the paper.

\section{Multivariate Hawkes Process and the Multivariate General Compound Hawkes Process}

The one-dimensional general compound Hawkes process is defined as
\[
S_t = S_0 + \sum_{k=1}^{N(t)} a (X_k),
\]
where $N(t)$ is an one-dimensional Hawkes process, $X_n$ is an ergodic continuous-time finite state (or countable state) Markov chain, independent of $N_t$, with space state $X$, and $a(x)$ is a bounded continuous function on $X$. 

Let $ \Vec{N}_t = (N_{1,t}, N_{2,t}, \cdots, N_{d,t},)$ be a $d$-dimensional Hawkes process with intensity function for each element $N_{i,t}$  
\begin{equation}
	\lambda_i(t) = \lambda_{\infty,i} + \int_{(0,t)} \sum_{j=1}^d \mu_{ij}(t-s) d N_{j,s},
\end{equation}
where $\lambda_{\infty,i} \in \mathbb{R}_+ $ and the intensity $\mu_{ij}(t)$ is a function from $\mathbb{R}_+$ to $\mathbb{R}_+$. Let $\boldsymbol{\mu} = (\mu_{ij})_{1 \leq i,j \leq d}$ and Let $\mathbf{K} = \int_0^{\infty} \boldsymbol{\mu} (t) dt$.

The MGCHP $\vec{S}_t = (S_{1,t}, S_{2,t}, \cdots, S_{d,t},)$ is defined as
\begin{equation}\label{MCHP}
	S_{i,t} = S_{i,0} + \sum_{k=1}^{N_{i,t}} a(X_{i,k}),
\end{equation}
where $X_{i,k}$ are independent ergodic continuous-time Markov chains.
\section{The Exponential Multivariate General Compound Hawkes Process model and Limit Theorems}
We introduce a new type of MGCHP, namely the exponential MGCHP (EMGCHP). Recall the definition of MGCHP $\vec{S}_t = (S_{1,t}, S_{2,t}, \cdots, S_{d,t})$
\begin{equation}\label{MGCHP_ch5}
	S_{i,t} =  S_{i,0} + \sum_{k=1}^{N_{i,t}} a(X_{i,k}),
\end{equation}
we can find the price process $\vec{S}_t$ is not always greater than $0$. In the numerical simulation with real trading data, the initial price $\vec{S}_0$ is much greater than the jump size $a(X_{i,k})$. So, we didn't have the non-positive stock prices in previous numerical examples. However, we need to fix this problem mathematically when we consider financial applications such as portfolio management problems and option pricing problems. Here, we introduce the exponential MGCHP (EMGCHP): 
\begin{definition}
	Let $ \Vec{N}_t = (N_{1,t}, N_{2,t}, \cdots, N_{d,t},)$ be a $d$-dimensional Hawkes process and $X_{i,k}$ are independent ergodic continuous-time Markov chains. We call $ \Vec{A}_t = e^{\Vec{S}_t}= (A_{1,t}, A_{2,t}, \cdots, A_{d,t},)$ is an exponential multivariate general compound Hawkes process (EMGCHP) if
	\begin{equation}\label{EMGCHP}
		A_{i,t}  =   \exp \left(S_{i,0} + \sum_{k=1}^{N_{i,t}} a(X_{i,k})\right).
	\end{equation}
	where $\vec{S}_t = (S_{1,t}, S_{2,t}, \cdots, S_{d,t},)$ is a MGCHP. And $A_{i,0} = \exp \left( S_{i,0} \right)$ is the initial price at time $t=0$ for the $i$th stock. 
\end{definition}

\begin{rem}
	In the MGCHP model, the Hawkes process $\vec{N}_t$ denotes the order flow at time $t$ and the Markov chain $a(X_{i,k})$ represents the change size for the stock price when an order arrives. So, the price changes by adding $a(X_{i,k})$. As for the EMGCHP model, the price changes by multiplying a factor $e^{a(X_{i,k})}$ when an order arrives. 	
\end{rem}



We focus on the law of large numbers and the functional central limit theorems for the EMGCHP in this section. Let $\circ$ denote the element-wise product, namely for matrices (or vector) $A$ and $B$, $(A\circ B)_{ij} = A_{ij}B_{ij}$. And let $\exp [\cdot]$ denote the element-wise exponential function, namely for matrix (or vector) $A$, $(\exp [A])_{ij} = \exp(A_{ij})$. Then we have the following limit theorems.

\begin{theorem}: \textbf{FCLT \RomanNumeralCaps{1} for EMGCHP: Stochastic Centralization} \label{FCLT_EMGCHP}
	Let $X_{i,k}, \, i=1,2,\cdots,d$ be independent ergodic Markov chains with $n$ states $\{1,2,\cdots,n\}$ and with ergodic probabilities $\left(\pi_{i,1}^{*}, \pi_{i,2}^{*}, \ldots, \pi_{i,n}^{*}\right)$. Let $\vec{A}_t$ be $d$-dimensional EMGCHP, we have 
	\begin{equation} \label{ETFCLT}
		 \left(\vec{A}_{nt}\right)^{\frac{1}{\sqrt{n}}} \circ \exp{ \left[- \frac{1}{\sqrt{n}} \Tilde{a}^* \vec{N}_{nt}\right]}  \longrightarrow \exp{ \left[\Tilde{\sigma}^* \boldsymbol{\Sigma^{1/2}} \vec{W}(t)\right] }, \, for\, all\, t>0 
	\end{equation}   
	as $n \rightarrow \infty$, where $\vec{W}(t)$ is a standard $d$-dimensional Brownian motion, $\mathbf{\Sigma}$ is a diagonal matrix such that $\mathbf{\Sigma}_{ii} = ((\mathbf{I}-\mathbf{K})^{-1}\vec{\lambda}_\infty)_i$, $\vec{N}_{nt}$ is a $d$-dimensional vector, $\Tilde{a}^*$ and $\Tilde{\sigma}^*$ are diagonal matrices:
	\[
	\Tilde{a}^* =
	\begin{bmatrix}
		a^*_1 & \cdots & 0\\
		\vdots& \ddots & \vdots\\
		0 & \cdots & a^*_d
	\end{bmatrix}, \,   \vec{N}_{nt} =
	\begin{bmatrix}
		N_{1,nt}  \\
		\vdots \\
		N_{d,nt}
	\end{bmatrix}, \,\Tilde{\sigma}^* =
	\begin{bmatrix}
		\sigma^*_1 &\cdots & 0 \\
		\vdots & \ddots & \vdots \\
		0 & \cdots & \sigma^*_d
	\end{bmatrix}.
	\]
	
	Here, $a_i^* = \sum_{k \in X_i} \pi_{i,k}^{*} a\left(X_{i,k}\right)$, and $\left(\sigma^{*}_i\right)^{2} :=\sum_{k \in X_i} \pi_{i,k}^{*} v_i(k)$ with
	\[
	\begin{aligned} v_i(k) &=b_i(k)^{2}+\sum_{j \in X_i}(g_i(j)-g_i(k))^{2} P_i(k, j) - 2 b_i(k) \sum_{j \in X_i}(g_i(j)-g_i(k)) P_i(k, j) \\ b_i &=(b_i(1), b_i(2), \ldots, b_i(n))^{\prime} \\ b_i(k) : &=a_i(k)-a_i^{*} \\ g_i : &=\left(P_i+\Pi_i^{*}-I\right)^{-1} b_i, \end{aligned}
	\]
	where $P_i$ is the  transition probability matrix for the Markov chain $X_i$, $\Pi_i^*$ is the matrix of stationary distributions of $P_i$,
	and $g_i(j)$ is the $j$th entry of $g_i$.	
\end{theorem}

\begin{proof}
	From the FCLT 1 for the MGCHP \cite{Guo2020}, we have 
	\begin{equation} 
		\frac{\vec{S}_{nt} - \Tilde{a}^* \vec{N}_{nt}}{\sqrt{n}} \xrightarrow{n \rightarrow \infty} \Tilde{\sigma}^* \boldsymbol{\Sigma^{1/2}} \vec{W}(t), \, for\, all\, t>0.
	\end{equation} 
	In this way for $i$th element in the vector $\vec{S}_{nt}$, we have
		\begin{equation} 
		\frac{{S}_{i,nt} - {a}_i^* {N}_{i,nt}}{\sqrt{n}} \xrightarrow{n \rightarrow \infty} {\sigma}_i^* \boldsymbol{(\Sigma)_{ii}^{1/2}} \vec{W}_i(t), \, for\, all\, t>0.
	\end{equation}

	Since exponential function is continues, we can derive
	\begin{equation} \label{elimit}
		\exp{  \left(   \frac{{S}_{i,nt} - {a}_i^* {N}_{i,nt}}{\sqrt{n}}  \right)   }
		 \xrightarrow{n \rightarrow \infty} \exp{\left({\sigma}_i^* \boldsymbol{(\Sigma)_{ii}^{1/2}} \vec{W}_i(t)\right)}   , \, for\, all\, t>0.
	\end{equation}  
	Rewriting Equation (\ref{elimit}) into a matrix form, we finish the proof.	
\end{proof}

\begin{theorem}\textbf{LLN for EMGCHP}\label{LLNEMGCHP}
	Let $\vec{S}_{nt}$ be a multivariate general compound Hawkes process defined before, we have
	\begin{equation} \label{llnemgchp}
		{  \left(\vec{A}_{nt} \right)^{\frac{1}{n}}        }	 \longrightarrow \exp{\left[ \Tilde{a}^* \boldsymbol{\Sigma} \vec{t} \right]}  
	\end{equation}	
	in the law of Skorohod topology as $n \rightarrow +\infty$, where $\Tilde{a}^*$ and $\boldsymbol{\Sigma}$ are defined before. $\vec{t}$ is a vector of time $t$.
\end{theorem}

\begin{proof}
	Recalling the LLN for MGCHP \cite{Guo2020}
	\begin{equation} \label{llnemgchp1}
		\frac{\vec{S}_{nt}}{n}  \xrightarrow{n \rightarrow \infty} \Tilde{a}^* \boldsymbol{\Sigma} \vec{t}.
	\end{equation}	
	Then, for each element in $\vec{S}_{nt}$, the LLN can be written as
	\begin{equation} 
		\frac{{S}_{i,nt}}{n} \xrightarrow{n \rightarrow \infty} {a}_i^* \boldsymbol{\Sigma}_{ii} t.
	\end{equation}.	
	Taking the exponential function for both sides, we 
	have
	\begin{equation} \label{llnemgchp2}
		\exp\left(\frac{{S}_{i,nt}}{n}\right) = { \left({A}_{i,nt}\right)^{\frac{1}{n}}}  \xrightarrow{n \rightarrow \infty} \exp { \left( {a}_i^* \boldsymbol{\Sigma}_{ii} t \right) }.
	\end{equation}.	
Rewriting (\ref{llnemgchp2}) into the matrix form we can derive the LLN for the EMGCHP.
\end{proof}

Similar as the MGCHP, we also consider a special case for the EMGCHP, namely the two-sate EMGCHP. Let the mid-price process  $ \Vec{A}_t = (A_{1,t}, A_{2,t}, \cdots, A_{d,t},)$ be
\begin{equation}\label{2dECHP}
	A_{i,t} = A_{i,0} \exp \left( \sum_{k=1}^{N_{i,t}} X_{i,k} \right)  ,
\end{equation}
where $X_{i,k}$ are independent ergodic Markov chains with two states $(+\delta, -\delta)$ and ergodic probabilities $(\pi_i^*, 1-\pi_i^*)$. Here, $\delta$ denotes the fixed tick size.

\begin{corollary}\textbf{FCLT \RomanNumeralCaps{1} for two-state EMGCHP: Stochastic Centralization} \label{cffcltemgchp}
	\begin{equation}
		 \left( \vec{A}_{nt}\right)^{\frac{1}{\sqrt{n}}}  \circ \exp {\left[ - \frac{1}{\sqrt{n}} \Tilde{a}^* \vec{N}_{nt} \right]}  \longrightarrow \exp \left[\Tilde{\sigma}^* \boldsymbol{\Sigma}^{1/2} \vec{W}(t)\right], \, for\, all\, t>0 
	\end{equation}   
	as $n \rightarrow \infty$, where $\vec{W}(t)$ is a standard $d$-dimensional Brownian motion, $\vec{N}_{nt}$ is a $d$-dimensional vector, $\Tilde{a}^*$ and $\Tilde{\sigma}^*$ are diagonal matrices
	\[
	\Tilde{a}^* =
	\begin{bmatrix}
		a^*_1 & \cdots & 0\\
		\vdots& \ddots & \vdots\\
		0 & \cdots & a^*_d
	\end{bmatrix}, \,   \vec{N}_{nt} =
	\begin{bmatrix}
		N_{1,nt}  \\
		\vdots \\
		N_{d,nt}
	\end{bmatrix}, \,\Tilde{\sigma}^* =
	\begin{bmatrix}
		\sigma^*_1 &\cdots & 0 \\
		\vdots & \ddots & \vdots \\
		0 & \cdots & \sigma^*_d
	\end{bmatrix}.
	\]
	Here $a_i^*:= \delta (2 \pi_i^* - 1)$, and let $(p_i, p_i')$ be transition probabilities  
	\begin{equation}\label{sigmastar_echp}
		{\sigma_i^*}^2:= 4\delta^2 \bigg( \frac{1 - p' + \pi^*(p' - p)}{ (p + p' -2)^2} - \pi^*(1 - \pi^*) \bigg).
	\end{equation}
\end{corollary}

\begin{corollary} \textbf{LLN for two-state EMGCHP}	\label{cfllnemgchp}
	
	Let $\vec{A}_{nt}$ be a EMGCHP defined before, we have
	\begin{equation} \label{lln2_ECHP}
		{ \left( \vec{A}_{nt} \right)^{\frac{1}{n}}}	 \longrightarrow \exp { \left[  \Tilde{a}^* \boldsymbol{\Sigma} \vec{t} \right]  } 
	\end{equation}	
	in the law of Skorohod topology as $n \rightarrow +\infty$, where $\Tilde{a}^*$ and $\boldsymbol{\Sigma}$ are defined in (\ref{sigmastar_echp}).
\end{corollary}

\begin{proof}
	Corollary \ref{cffcltemgchp} and \ref{cfllnemgchp} can be proved directly from the Theorem \ref{FCLT_EMGCHP} and Theorem \ref{LLNEMGCHP} by setting the Markov chain $X_{i,k}$ with two states $(+\delta, -\delta)$  and function $a(x) = x$.
\end{proof}

The FCLT \RomanNumeralCaps{1} of EMGCHP has a similar disadvantage as the FCLT \RomanNumeralCaps{1} of MGCHP. It cannot provide us a model for prediction tasks since we cannot observe the order flow $\vec{N}_t$ in advance. This motivates us to consider a pure diffusion limits of the EMGCHP.

\begin{theorem}: \textbf{FCLT \RomanNumeralCaps{2} for EMGCHP: Deterministic Centralization} \label{EFDFCLT}
	
	Let $X_{i,k}, \, i=1,2,\cdots,d$ be independent ergodic Markov chains with $n$ states $\{1,2,\cdots,n\}$ and with ergodic probabilities $\left(\pi_{i,1}^{*}, \pi_{i,2}^{*}, \ldots, \pi_{i,n}^{*}\right)$. Let $\vec{A}_{nt}$ be $d$-dimensional EMGCHP, we have 
	
	\begin{equation} \label{EFFCLT}
		 \left(\vec{A}_{nt}\right)^{\frac{1}{\sqrt{n}}} \circ \exp { \left[ -\Tilde{a}^* n t (\mathbf{I}-\mathbf{K})^{-1}\vec{\lambda}_\infty \right] } \xrightarrow{n \rightarrow \infty} \exp{   \left[ \Tilde{\sigma}^* \boldsymbol{\Sigma^{1/2}} \vec{W}_1(t) + \Tilde{a}^* (\mathbf{I}-\mathbf{K})^{-1} \boldsymbol{\Sigma^{1/2}} \vec{W}_2(t)\right]}
	\end{equation}  
	for all $t>0$, where $\vec{W}_1(t)$ and $\vec{W}_2(t)$ are independent standard $d$-dimensional Brownian motions. Other parameters $\Tilde{a}^*$, $\Tilde{\sigma}^*$, $\boldsymbol{\Sigma^{1/2}}$, $\vec{\lambda}_\infty $, and $\mathbf{K}$ are defined before.
\end{theorem}

\begin{proof}
	From \cite{Guo2020}, we have the FCLT 2 for MGCHP in the form of
	\begin{equation} \label{FFCLT2020}
		\frac{\vec{S}_{nt} - \Tilde{a}^* n t (\mathbf{I}-\mathbf{K})^{-1}\vec{\lambda}_\infty}{\sqrt{n}} \xrightarrow{n \rightarrow \infty} \Tilde{\sigma}^* \boldsymbol{\Sigma^{1/2}} \vec{W}_1(t) + \Tilde{a}^* (\mathbf{I}-\mathbf{K})^{-1} \boldsymbol{\Sigma^{1/2}} \vec{W}_2(t), \, for\, all\, t>0.
	\end{equation}
	Taking the exponential function for each element in the left-hand side, we derive the following limit 
	\begin{equation} \label{EEFFCLT}
		\exp \left[ \frac{\vec{S}_{nt} - \Tilde{a}^* n t (\mathbf{I}-\mathbf{K})^{-1}\vec{\lambda}_\infty}{\sqrt{n}} \right] \xrightarrow{n \rightarrow \infty} \exp \left[ \Tilde{\sigma}^* \boldsymbol{\Sigma^{1/2}} \vec{W}_1(t) + \Tilde{a}^* (\mathbf{I}-\mathbf{K})^{-1} \boldsymbol{\Sigma^{1/2}} \vec{W}_2(t)\right].
	\end{equation}
	Rewrite Equation ($\ref{EEFFCLT}$) into ($\ref{EFFCLT}$), we finish the proof.
\end{proof}

With the FCLT \RomanNumeralCaps{2}, we consider an approximation: 
\begin{equation} \label{FCLTap2emgchp}
	\vec{A}_{nt}  \sim \exp \left[ \sqrt{n} \Tilde{\sigma}^* \boldsymbol{\Sigma^{1/2}} \Vec{W}_1(t) + \sqrt{n} \Tilde{a}^* (\mathbf{I}-\mathbf{K})^{-1} \boldsymbol{\Sigma^{1/2}} \Vec{W}_2(t) +   \Tilde{a}^* n  (\mathbf{I}-\mathbf{K})^{-1}\vec{\lambda}_\infty t \right], 
\end{equation} 
for all $t>0$ and some large enough $n$. Here, $\Vec{W}_1(t)$ and $\Vec{W}_2(t)$ are two independent standard $d$-dimensional Brownian motions. Let 
\[C = \left( \sqrt{n} \Tilde{\sigma}^* \boldsymbol{\Sigma^{1/2}},  \sqrt{n} \Tilde{a}^* (\mathbf{I}-\mathbf{K})^{-1} \boldsymbol{\Sigma^{1/2}} \right) \in \mathbb{R}^{d \times 2d} \] 
and $\vec{W}_t$ be a $2d$-dimensional standard Brownian motion, then 
\[\vec{A}_{nt} = \exp { \left[ C\vec{W}_t + \Tilde{a}^* n  ((\mathbf{I}-\mathbf{K})^{-1}\vec{\lambda}_\infty)t \right]  }.\]
By applying the multivariate It\^{o} formula, we derive the price process in the form a multivariate geometric Brownian motion  
\begin{equation} \label{fclt_app}
	dA_{i,t} = A_{i,t} \left\{(  a^*_in  ((\mathbf{I}-\mathbf{K})^{-1}\vec{\lambda}_\infty)_i + \frac{1}{2}(CC')_{ii} ) dt +  C_i d\vec{W}_t \right\},
\end{equation} 
where $C_i$ is the $i$th row of $C = \left(C_1,C_2,\cdots, C_d\right)$.

\subsection{Special Case: One-dimensional EMGCHP and its Limit Theorems}

We consider the one-dimensional EMGCHP and its limit theorems in this section. Let $N_t$ be a one-dimensional Hawkes process with intensity in the form of
\[
\lambda(t) = \lambda_{\infty} + \int_{0}^{t} \mu (t-s) dN_s,
\]
then the one-dimensional EMGCHP is defined as
\begin{equation}\label{1demgchp}
A_{t} = A_{0} \exp \left( \sum_{k=1}^{N_{t}} X_{k} \right).
\end{equation}
 
Two FCLTs and LLN can be derived directly from the multivariate case.

\begin{theorem}: \textbf{FCLT \RomanNumeralCaps{1} for one-dimensional EMGCHP: Stochastic Centralization} 
	Let $X_{k}$ be ergodic Markov chains with $n$ states $\{1,2,\cdots,n\}$ and with ergodic probabilities $\left(\pi_{1}^{*}, \pi_{2}^{*}, \ldots, \pi_{n}^{*}\right)$. Let ${A}_t$ be $1$-dimensional EMGCHP, we have 
	\begin{equation} 
	{  \left({A}_{nt}\right)^{\frac{1}{\sqrt{n}}}    e^{-\frac{1}{\sqrt{n}}{a}^* {N}_{nt}}}  \longrightarrow e^{ {\sigma}^* {\sqrt{\frac{\lambda_{\infty}}{1-\hat{\mu}}}} {W}(t)}, \, for\, all\, t>0 
	\end{equation}   
	as $n \rightarrow \infty$, where ${W}(t)$ is a standard $1$-dimensional Brownian motion, $\hat{\mu}=\int_{0}^{\infty}\mu(s)ds$. $a^*$ and $\sigma^*$ are defined below:	
	$a^* = \sum_{k \in X} \pi_{k}^{*} a\left(X_{k}\right)$, and $\left(\sigma^{*}\right)^{2} :=\sum_{k \in X} \pi_{k}^{*} v(k)$ with
	\[
	\begin{aligned} v(k) &=b(k)^{2}+\sum_{j \in X}(g(j)-g(k))^{2} P(k, j) - 2 b(k) \sum_{j \in X}(g(j)-g(k)) P(k, j) \\ b &=(b(1), b(2), \ldots, b(n))^{\prime} \\ b(k) : &=a(k)-a^{*} \\ g : &=\left(P+\Pi^{*}-I\right)^{-1} b, \end{aligned}
	\]
	where $P$ is the  transition probability matrix for the Markov chain $X$, $\Pi^*$ is the matrix of stationary distributions of $P$,
	and $g(j)$ is the $j$th entry of $g$.	
\end{theorem} 
 
\begin{theorem}\textbf{LLN for One-dimensional EMGCHP}
	Let ${A}_{nt}$ be an one-dimensional EMGCHP, we have
	\begin{equation} 
	\sqrt[n]{{A}_{nt}}	 \longrightarrow e^{{a}^* {\frac{\lambda_{\infty}}{1-\hat{\mu}}} t }
	\end{equation}	
	in the law of Skorohod topology as $n \rightarrow +\infty$, where ${a}^*$ is defined before.
\end{theorem}

\begin{theorem}: \textbf{FCLT \RomanNumeralCaps{2} for One-dimensional EMGCHP: Deterministic Centralization} 
	
	Let $X_{k}$ be ergodic Markov chains with $n$ states $\{1,2,\cdots,n\}$ and with ergodic probabilities $\left(\pi_{1}^{*}, \pi_{2}^{*}, \ldots, \pi_{n}^{*}\right)$. Let ${A}_{nt}$ be $1$-dimensional EMGCHP, we have 
	
	\begin{equation}
	{ \left({A}_{nt}\right)^{\frac{1}{\sqrt{n}}}  e^{-{a}^* \sqrt{n} t {\frac{\lambda_{\infty}}{1-\hat{\mu}}}}} \xrightarrow{n \rightarrow \infty} e^{{\sigma}^* {\sqrt{\frac{\lambda_{\infty}}{1-\hat{\mu}}}} {W}_1(t) + {a}^* {\sqrt{\frac{\lambda_{\infty}}{(1-\hat{\mu})^3}}} {W}_2(t)}, \, for\, all\, t>0 
	\end{equation}  
	where ${W}_1(t)$ and ${W}_2(t)$ are independent standard Brownian motions. Other parameters ${a}^*$ and ${\sigma}^*$ are defined before.
\end{theorem}

With the FCLT \RomanNumeralCaps{1} for the one-dimensional EMGCHP, we consider an approximation: 
\begin{equation}
{A}_{nt}  \sim \exp \left( \sqrt{n} {\sigma}^* {\sqrt{\frac{\lambda_{\infty}}{1-\hat{\mu}}}} {W}_1(t) + \sqrt{n} {a}^* {\sqrt{\frac{\lambda_{\infty}}{(1-\hat{\mu})^3}}} {W}_2(t) + {a}^*n  {\frac{\lambda_{\infty}}{1-\hat{\mu}}}t \right), 
\end{equation} 
for all $t>0$ and some large enough $n$. Noting that $\hat{\mu}=\int_{0}^{\infty}\mu(s)ds$. 

We can merge ${W}_1(t)$ and ${W}_2(t)$ in the form of 
\begin{equation}
{A}_{nt}  \sim \exp \left( 
\sqrt{n{\sigma^*}^2 \frac{\lambda_\infty}{1-\hat{\mu}} + n{a^*}^2 \frac{\lambda_{\infty}}{(1-\hat{\mu})^3}} W(t)+
{a}^*n  {\frac{\lambda_{\infty}}{1-\hat{\mu}}}t \right). 
\end{equation}

With It\^{o} formula, we can rewrite it as
\begin{equation} \label{1demgchpgeo}
\begin{aligned}
dA_{nt} &= A_{nt} \left(  a^*n \frac{\lambda_\infty}{1-\hat{\mu}} + n{\sigma^*}^2 \frac{\lambda_\infty}{2(1-\hat{\mu})} + n{a^*}^2 \frac{\lambda_{\infty}}{2(1-\hat{\mu})^3}\right) dt \\& + A_{nt} \sqrt{n{\sigma^*}^2 \frac{\lambda_\infty}{1-\hat{\mu}} + n{a^*}^2 \frac{\lambda_{\infty}}{(1-\hat{\mu})^3}} dW_t 
\end{aligned},
\end{equation}
if we consider the exponential decay as the intensity function, $\hat{\mu} = \alpha/\beta$.

\section{Option Pricing with MGCHP in LOB}

In 2019, Simard and R{\'e}millard considered the European option pricing in the LOB by a discrete-time model \cite{SR}. We consider the option pricing problem in the LOB with the EMGCHP model in this section, which is a continuous-time method. The FCLTs and LLN tell us that the stock price will converge to geometric Brownian motions if $n$ is large enough. In the high-frequency trading problem, orders are counted in milliseconds, which makes option expiry time relatively large. In this way, the approximation is appropriate for the option pricing problem. The convergence of options prices comes from the convergence of the processes, together with the uniform integrability of the options payoff. The precise result is that if $X_n$ converges in law to $X$ and $\Phi$ is uniformly integrable, then $E( \Phi(X_n))$  converges to $E(\Phi(X))$.

\subsection{Hawkes-based Black–Scholes Equation}

Let $A_{t}$ be the stock price at time $t$ in high-frequency context modeled by the EMGCHP in equation (\ref{EMGCHP}). With the FCLT II for the EMCGHP and approximation (\ref{FCLTap2emgchp}), we have the following price dynamics when $n$ is large enough 

\begin{equation} \label{sde_gbm_1}
	\begin{aligned}
		dA_{nt} &= A_{nt} \left(  a^*n \frac{\lambda_\infty}{1-\alpha/\beta} + n{\sigma^*}^2 \frac{\lambda_\infty}{2(1-\alpha/\beta)} + n{a^*}^2 \frac{\lambda_{\infty}}{2(1-\alpha/\beta)^3}\right) dt \\& + A_{nt} \sqrt{n{\sigma^*}^2 \frac{\lambda_\infty}{1-\alpha/\beta} + n{a^*}^2 \frac{\lambda_{\infty}}{(1-\alpha/\beta)^3}} dW_t 
	\end{aligned},
\end{equation}
where $\lambda_\infty$, $\alpha$, and $\beta$ are parameters in the Hawkes process with exponential decay excitation. $\sigma^*$ and $a^*$ are given in Theorem \ref{FCLT_EMGCHP}. Here we considered the $1$-dimensional EMGCHP, so we use notations $\sigma^*$ and $a^*$ instead of $\Tilde{\sigma}^*$ and $\Tilde{a}^*$. 

The stochastic differential equation (\ref{sde_gbm_1}) implies that $A_t$ is a geometric Brownian motion with drift $\left(  a^*n \frac{\lambda_\infty}{1-\alpha/\beta} + n{\sigma^*}^2 \frac{\lambda_\infty}{2(1-\alpha/\beta)} + n{a^*}^2 \frac{\lambda_{\infty}}{2(1-\alpha/\beta)^3}\right)$ and volatility $\sqrt{n{\sigma^*}^2 \frac{\lambda_\infty}{1-\alpha/\beta} + n{a^*}^2 \frac{\lambda_{\infty}}{(1-\alpha/\beta)^3}}$. By applying the Girsanov theorem, we can rewrite (\ref{sde_gbm_1}) as
\[
dA_t = rA_t dt + A_t \sqrt{n{\sigma^*}^2 \frac{\lambda_\infty}{1-\alpha/\beta} + n{a^*}^2 \frac{\lambda_{\infty}}{(1-\alpha/\beta)^3}} d \tilde{W}_t,
\]    
where 
\[
d\tilde{W}_t = d W_t + \frac{\left(  a^*n \frac{\lambda_\infty}{1-\alpha/\beta} + n{\sigma^*}^2 \frac{\lambda_\infty}{1-\alpha/\beta} + n{a^*}^2 \frac{\lambda_{\infty}}{(1-\alpha/\beta)^3} - r\right)}{\sqrt{n{\sigma^*}^2 \frac{\lambda_\infty}{1-\alpha/\beta} + n{a^*}^2 \frac{\lambda_{\infty}}{(1-\alpha/\beta)^3}}} dt
\]
is a Brownian motion under the risk-neutral measure.

Consider the vanilla European call option with payoff $(A_T - K)^+$, where $T$ is the maturity time and $K$ is the strike price. Let $c(t,x)$ denote the option price at time $t$ and $A_t = x$, $r$ denote the constant interest rate, we can derive the corresponding Black–Scholes partial differential equation in the form of
\begin{equation}\label{bsmhawkes}
	\begin{cases}
		c_t(t,x) + rxc_x(t,x)+\frac{1}{2}x^2\left(n{\sigma^*}^2 \frac{\lambda_\infty}{1-\alpha/\beta} + n{a^*}^2 \frac{\lambda_{\infty}}{(1-\alpha/\beta)^3}\right)c_{xx}(t,x) = rc(t,x)\\ c(T,x) = (x-K)^+\\
		c(t,0) = 0 \\
		\lim_{x \rightarrow \infty}\left[c(t,x)-(x-e^{-r(T-t)}K) \right] =0
	\end{cases}.
\end{equation}

According to the classical Black–Scholes equation theory (a good text book is \cite{Shr2004}), we have the solution of equation (\ref{bsmhawkes}) in the form of


\begin{equation}
	\vec{A}_{nt}  \sim \exp \left[ \sqrt{n} \Tilde{\sigma}^* \boldsymbol{\Sigma^{1/2}} \Vec{W}_1(t) + \sqrt{n} \Tilde{a}^* (\mathbf{I}-\mathbf{K})^{-1} \boldsymbol{\Sigma^{1/2}} \Vec{W}_2(t) +   \Tilde{a}^* n  (\mathbf{I}-\mathbf{K})^{-1}\vec{\lambda}_\infty t \right], 
\end{equation}

where 
\begin{equation}\label{d_pm}
	d_{\pm}(\tau,x) = \frac{1}{\sqrt{n\tau{\sigma^*}^2 \frac{\lambda_\infty}{1-\alpha/\beta} + n\tau{a^*}^2 \frac{\lambda_{\infty}}{(1-\alpha/\beta)^3}}} \left[\log \frac{x}{K} +  \left(r\pm  (\frac{n{\sigma^*}^2}{2} \frac{\lambda_\infty}{1-\alpha/\beta} + \frac{n{a^*}^2}{2} \frac{\lambda_{\infty}}{(1-\alpha/\beta)^3})\right)\tau  \right],
\end{equation}
and $\Phi$ denotes the cumulative standard normal distribution

\begin{equation}
	\Phi(y) = \frac{1}{\sqrt{2\pi}}\int_{-\infty}^{y} e^{-\frac{z^2}{2}}dz.
\end{equation}

\subsection{Numerical Examples}

Next, we consider the numerical examples for the Hawkes-based European call option. We take the parameters' values in previous chapter for Hawkes process and set the strike price $K = 50$, interest rate $r = 0.06$, maturity time $T=1$. Parameters can be found in Table \ref{option} and Figure \ref{option_ST_h} shows the relationship between option price $c(t,x)$, spot price $A_t=x$, and time $t$.  

\begin{table}[H]
	\centering
	\caption{Parameters for  Hawkes-based European call option}
	\begin{tabular}{ccccccccc}
		\toprule
		\textbf{Parameter} & $K$ & $r$ & $\alpha$ & $\beta$ & $\lambda_\infty$ &$\sigma^*$& $a^*$ &$T$    \\ \midrule
		Value   & 50   & 0.06   & 0.7     & 1 & 0.75 &0.05&0.03&1 \\ \bottomrule
	\end{tabular}
	\label{option}
\end{table}

\begin{figure}[H]
	\centering
	\includegraphics[width=10cm]{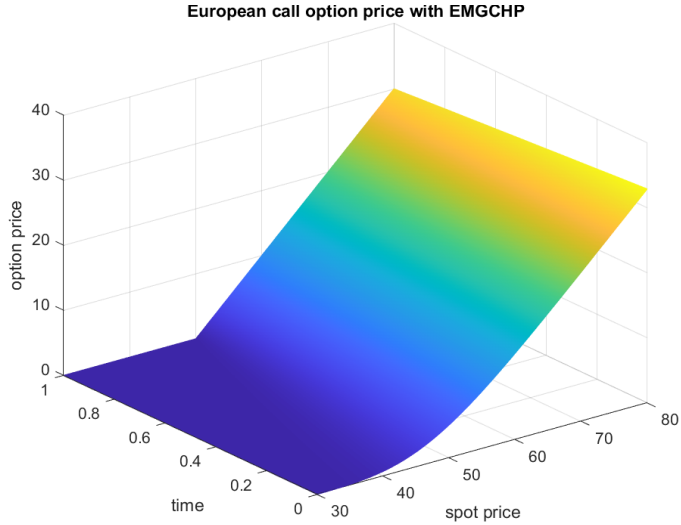}
	\caption{ Hawkes-based European call option.}
	\label{option_ST_h}
\end{figure}

We present some figures of the option price $c$ as a function of $\alpha/\beta$. Parameters can be found in Table \ref{option} and $\alpha/\beta \in [0.1,0.9]$. As can be seen in Figure \ref{option_ST_k}, for fixed spot price, when $\alpha/\beta$ is small, different attitude of the order flow's history cannot change the option price a lot. When $\alpha/\beta$ closes to $1$, the option price increases dramatically, as $\alpha/\beta$ increases.

\begin{figure}[H]
	\centering
	\begin{minipage}{0.4\textwidth}
		\includegraphics[width=\linewidth]{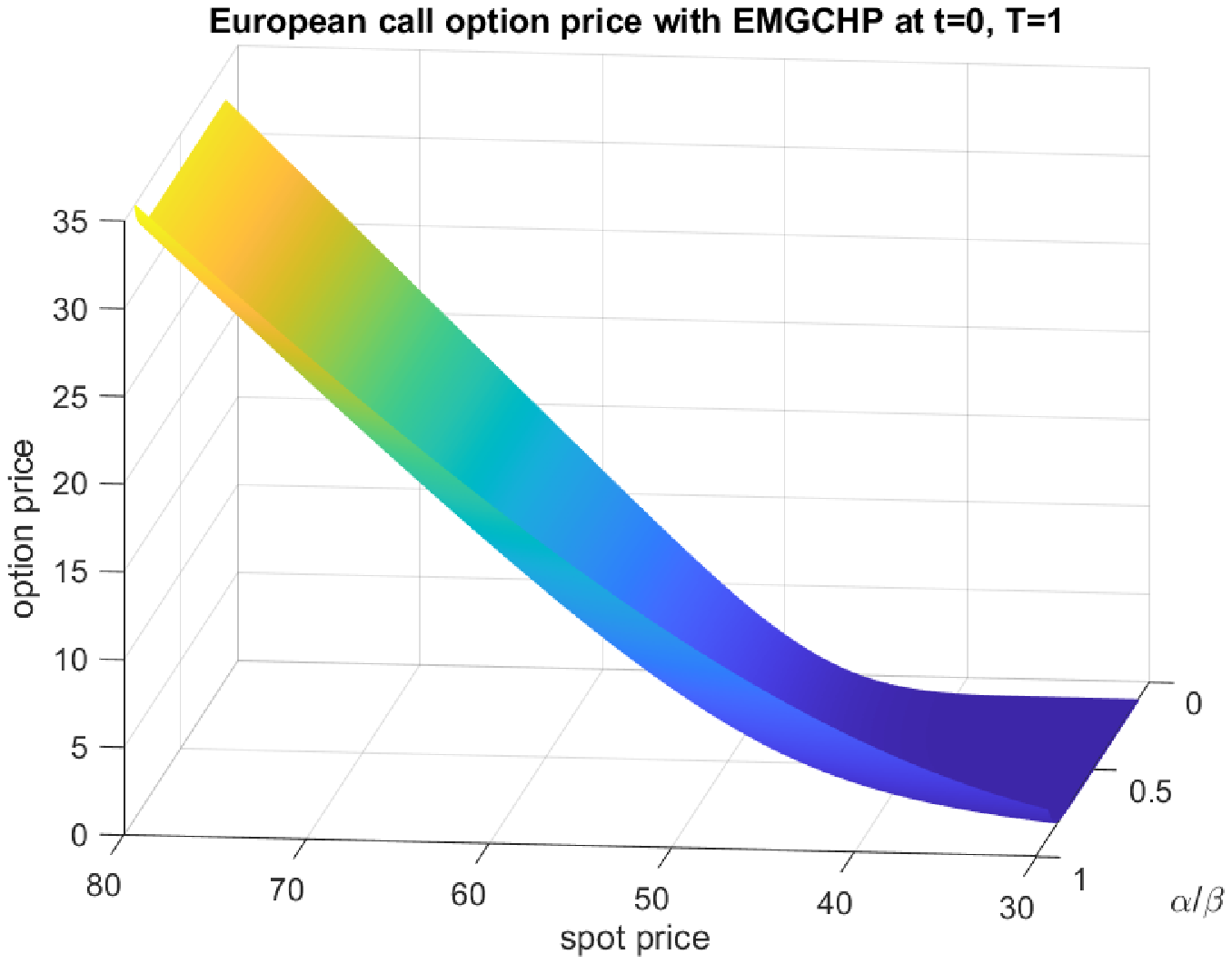}
	\end{minipage}
	\hspace{3mm} 
	\begin{minipage}{0.4\textwidth}
		\includegraphics[width=\linewidth]{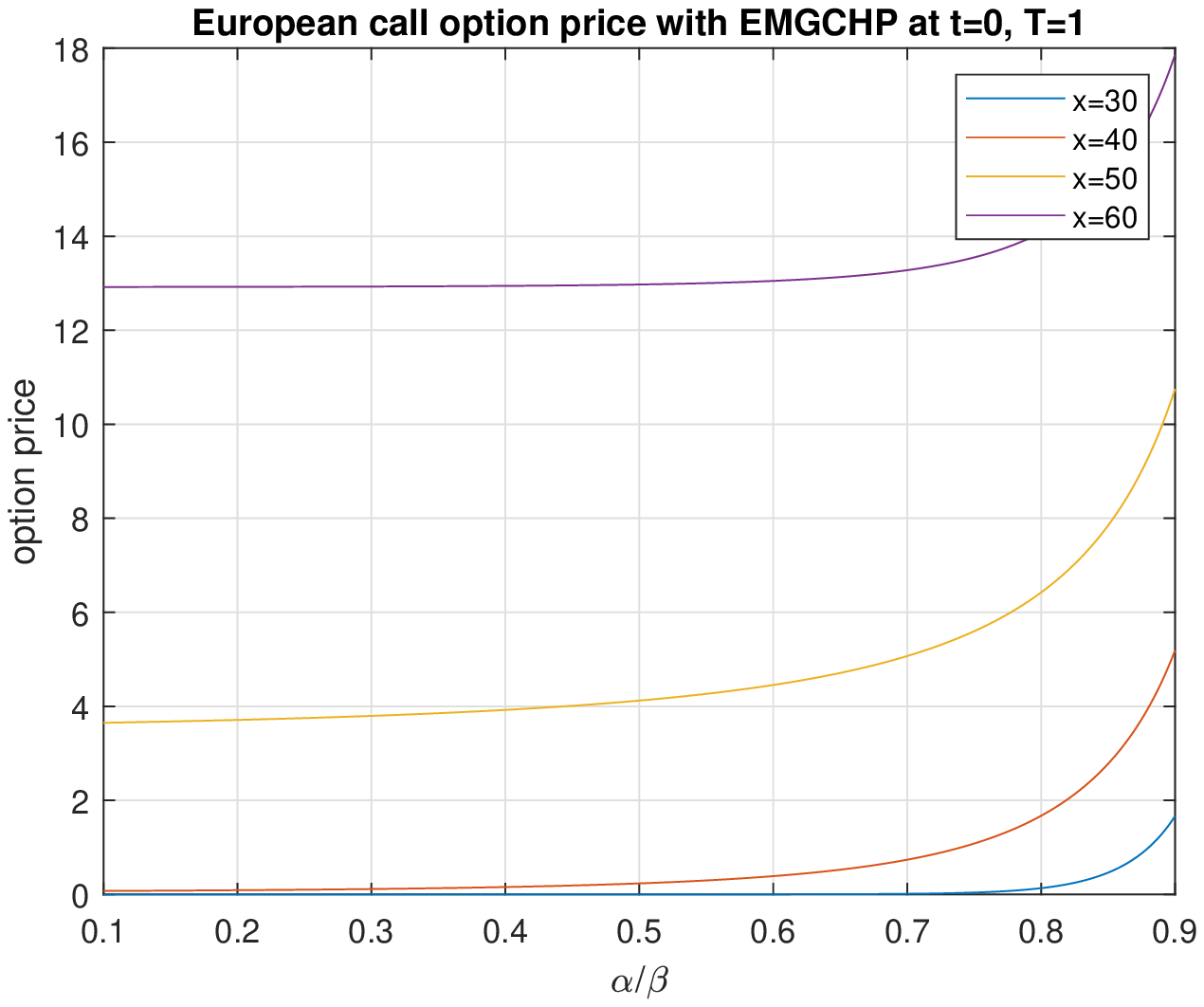}
	\end{minipage}
	\hspace{3mm} 
	\begin{minipage}{0.4\textwidth}
		\includegraphics[width=\linewidth]{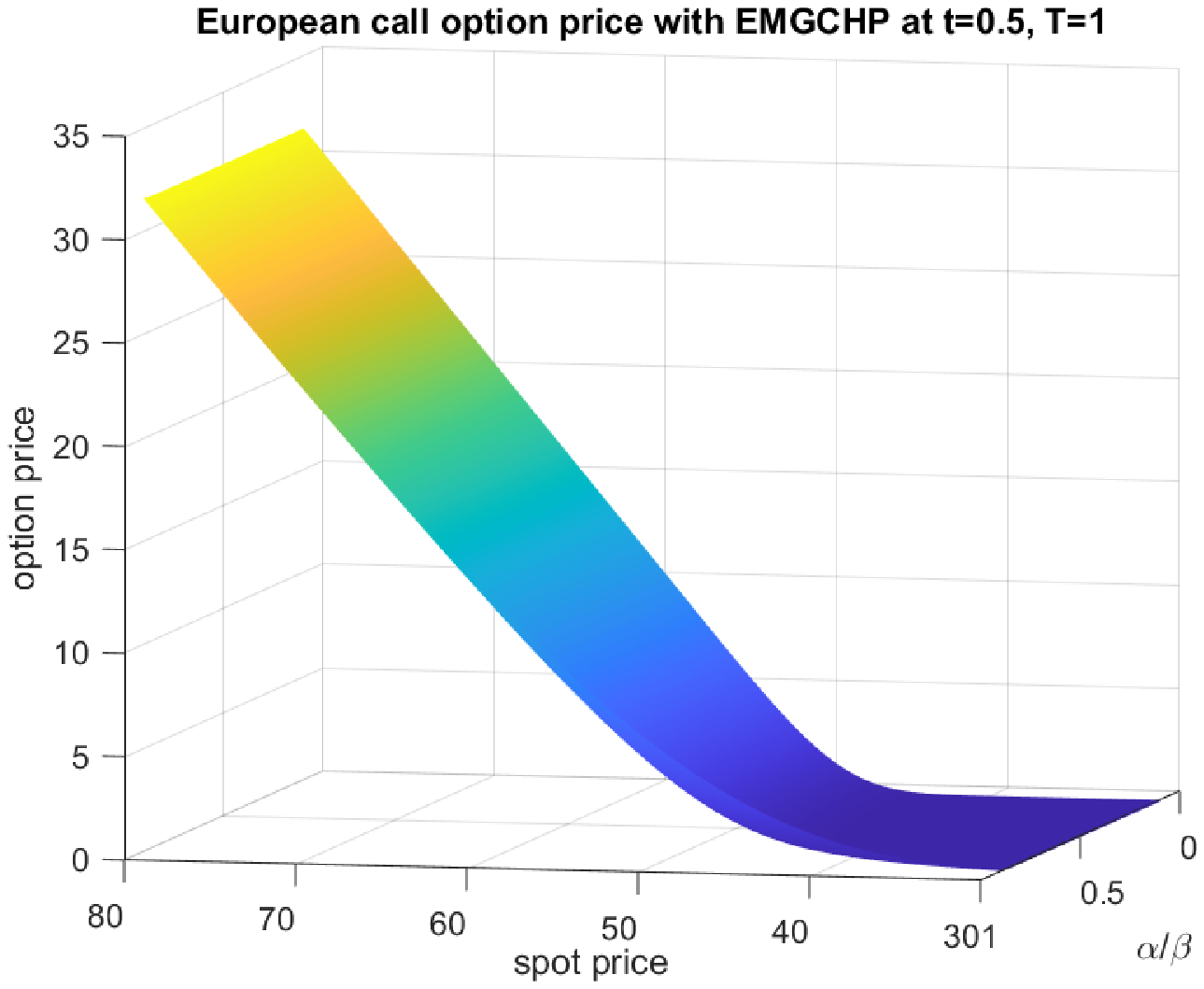}
	\end{minipage}
	\begin{minipage}{0.4\textwidth}
		\includegraphics[width=\linewidth]{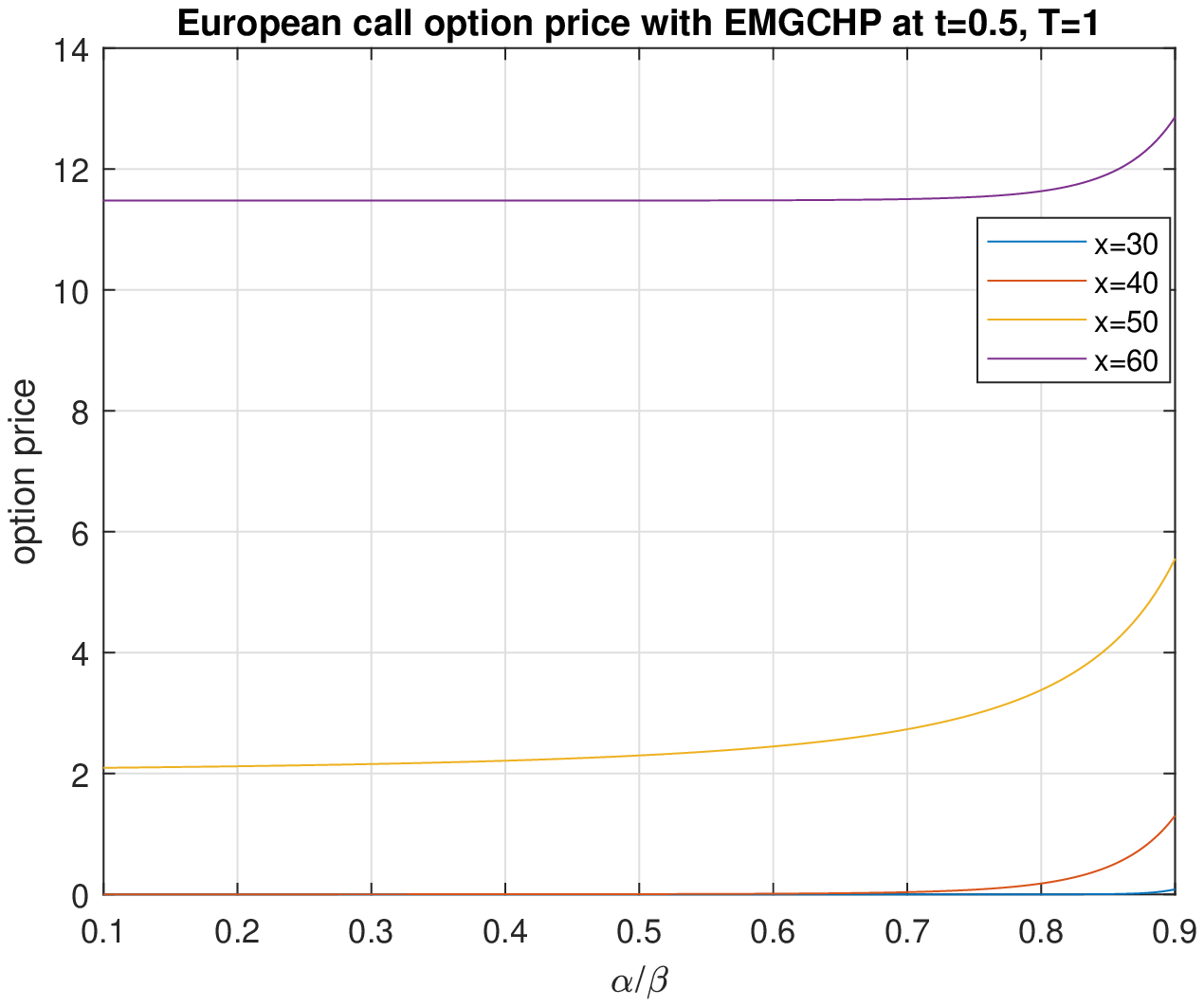}
	\end{minipage}
	\caption{Hawkes-based European call option with $\alpha/\beta \in [0.1,0.9]$ }
	\label{option_ST_k}
\end{figure}

Relationships of $\sigma^*$, $a^*$ and the option price $c$ can be found in Figure \ref{option_ST_sig}. We can conclude that the option price $c$ is proportional to the expectation of stock's jump size $a^*$ and the variance of stock's jump size $\sigma^*$.     

\begin{figure}[H]
	\centering
	\begin{minipage}{0.4\textwidth}
		\includegraphics[width=\linewidth]{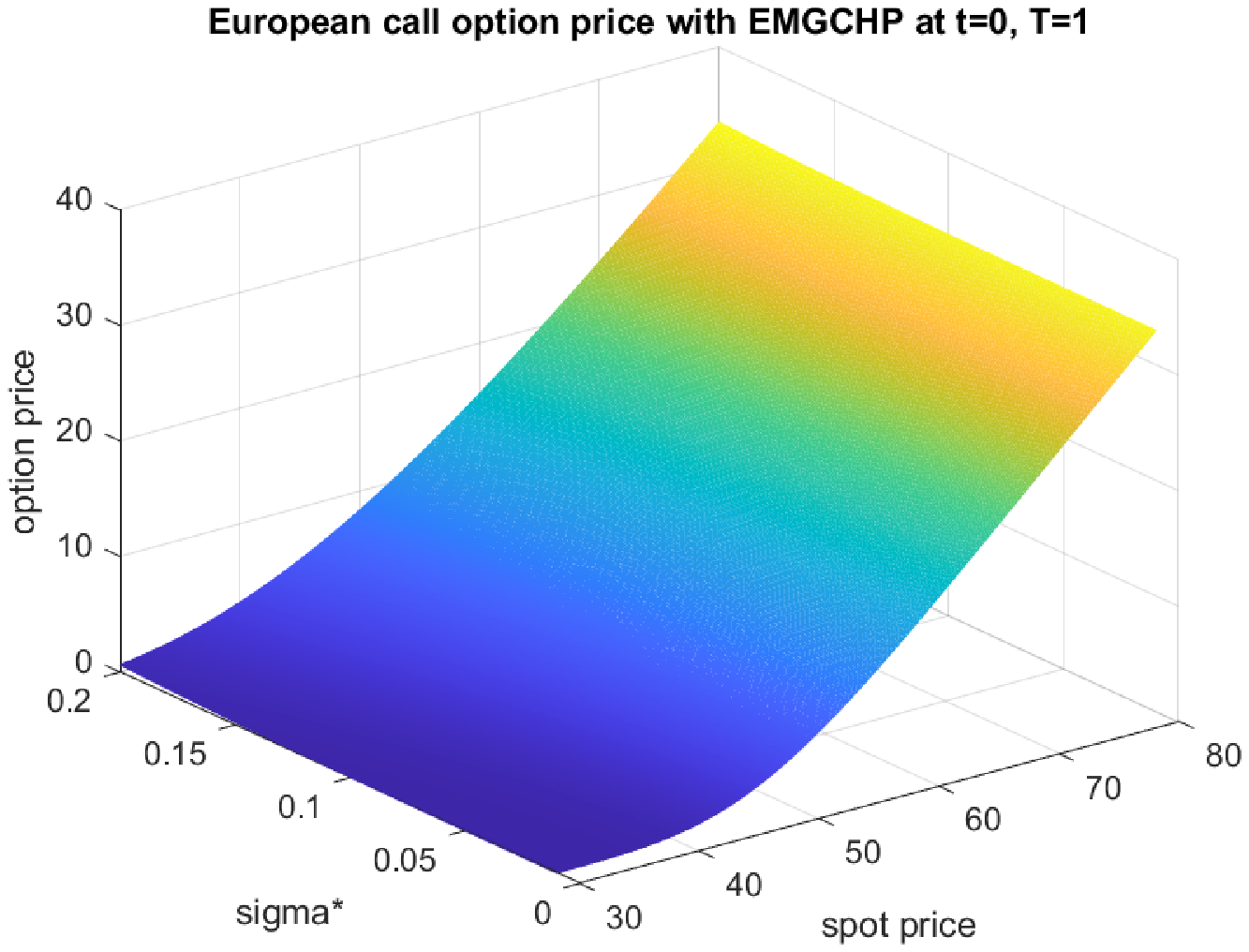}
	\end{minipage}
	\hspace{3mm} 
	\begin{minipage}{0.4\textwidth}
		\includegraphics[width=\linewidth]{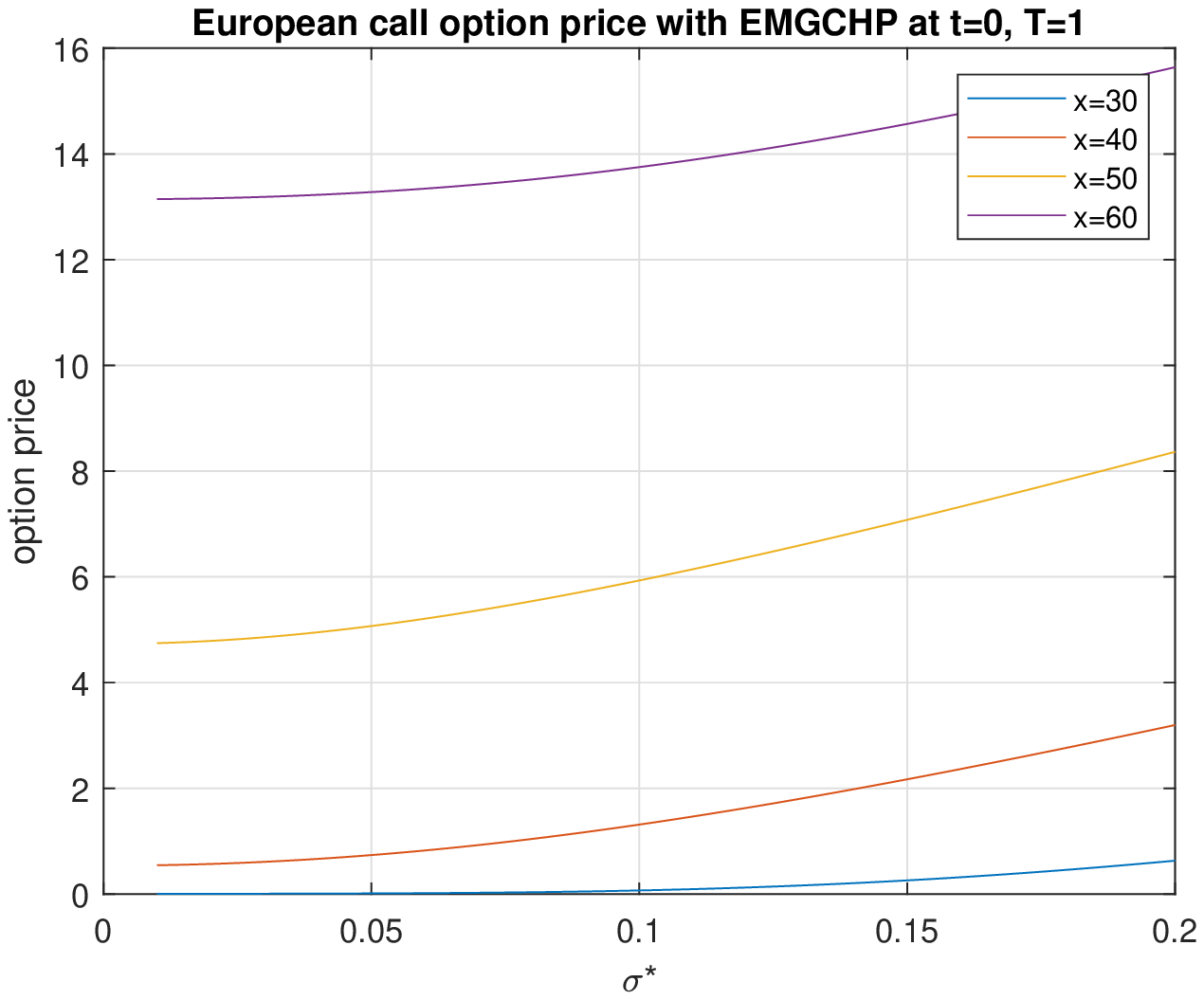}
	\end{minipage}
	\hspace{3mm} 
	\begin{minipage}{0.4\textwidth}
		\includegraphics[width=\linewidth]{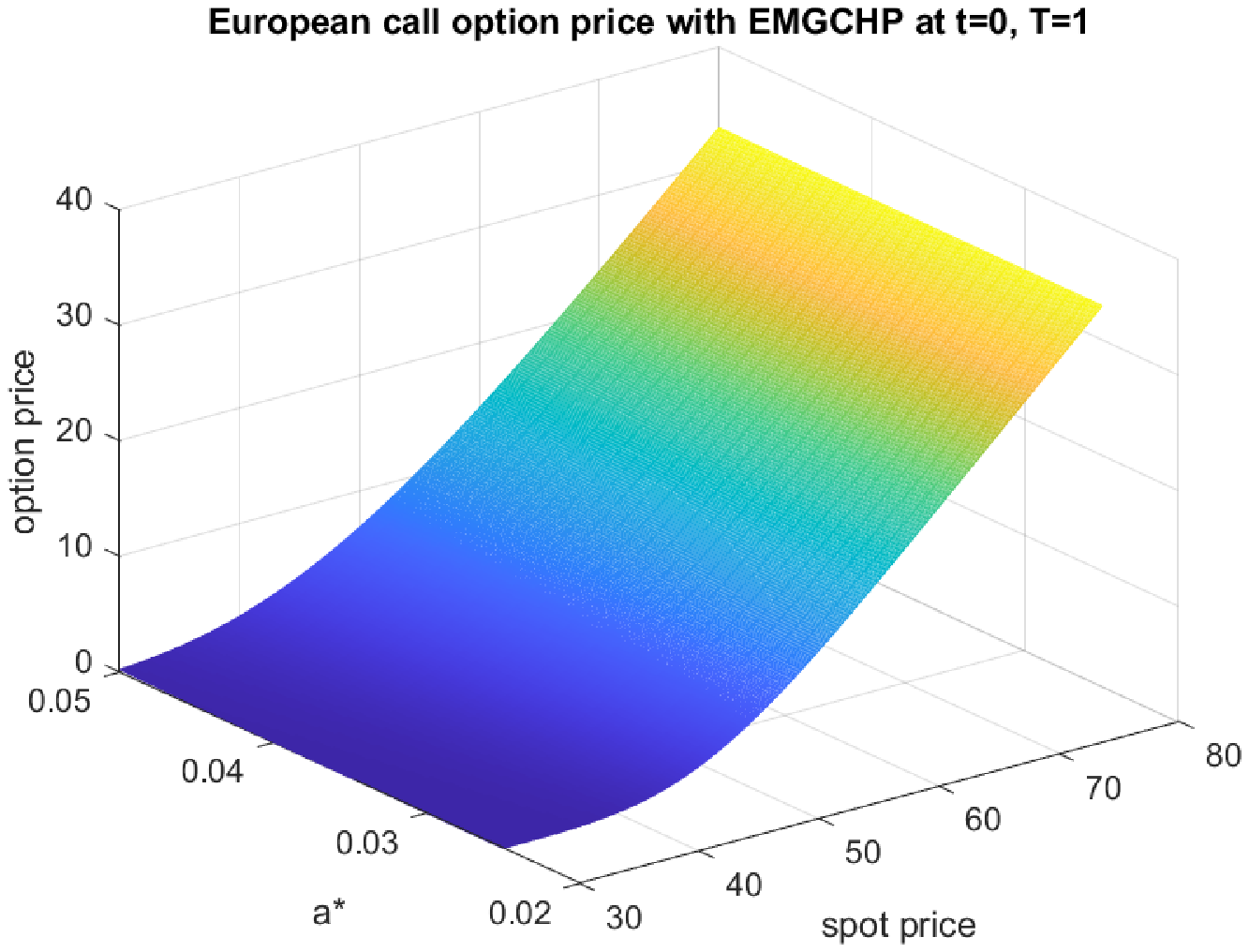}
	\end{minipage}
	\begin{minipage}{0.4\textwidth}
		\includegraphics[width=\linewidth]{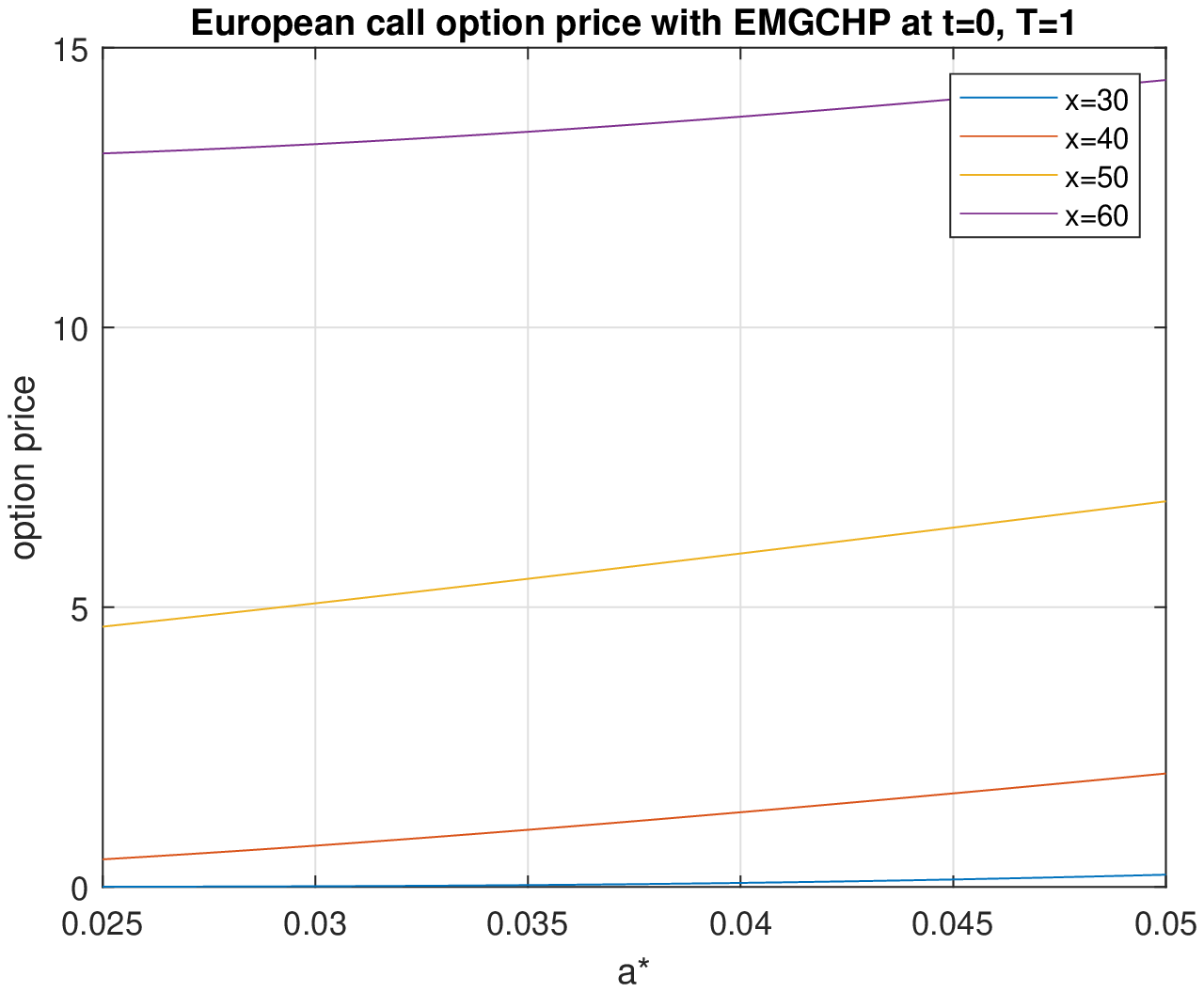}
	\end{minipage}
	\caption{Hawkes-based European call option with $\sigma^* \in [0.01,0.2]$ and $a^* \in [0.025,0.05]$.}
	\label{option_ST_sig}
\end{figure}



\subsection{Greeks in Hawkes-based Black–Scholes Model}

Greeks are partial derivatives of the option price $c(t,x)$. They denote the sensitivity of the option price to a change in underlying parameters. For example, delta $c_x(t,x)$ is the partial derivative with respect to the spot price $A(t) = x$ and theta  $c_t(t,x)$ is the partial derivative with respect to time $t$. We can also compute the corresponding Greeks for our Hawkes-based option pricing model. 

The delta for the Hawkes-based Black–Scholes model is in the form of 
\[c_x(t,x) = \Phi(d_{+}(T-t),x)\] 
and the theta is defined as
\[
c_t(t,x) = -rKe^{-r(T-t)}\Phi (d_{-}(T-t,x)) - \frac{x}{2\sqrt{T-t}} \phi (d_{+}(T-t,x)) \sqrt{n{\sigma^*}^2 \frac{\lambda_\infty}{1-\alpha/\beta} + n{a^*}^2 \frac{\lambda_{\infty}}{(1-\alpha/\beta)^3}}.
\] 
Since $\Phi$ and $\phi$ are always non-negative,delta and theta are always positive. Other classical Greeks are same as the Greeks in the classical Black-Scholes theory. More discussions can be found in John Hull's book \cite{Hul2012}. Here we conclude some special Greeks in the Hawkes-based model. 

$c_{\sigma^*}$ and $c_{a^*}$ are partial derivatives to measure the sensitivity of the Markov chain (the jump size of stock price).  
\begin{equation}
	\begin{aligned}
		c_{\sigma^*} (T-\tau,x)= x\phi (d_+(\tau,x)) F_+(\tau,x) - Ke^{-r\tau}\phi(d_-(\tau,x))F_-(\tau,x), 
	\end{aligned}
\end{equation}
\begin{equation}
	\begin{aligned}
		c_{a^*} (T-\tau,x)= x\phi (d_+(\tau,x)) G_+(\tau,x) - Ke^{-r\tau}\phi(d_-(\tau,x))G_-(\tau,x), 
	\end{aligned}
\end{equation}
where
\begin{equation}
	\begin{aligned}
		F_\pm (\tau,x) = - \left(1  +  \frac{{a^*}^2}{{\sigma^*}^2} \frac{1}{(1-\alpha/\beta)^2}\right)^{-1}d_\pm(\tau,x)  \pm \left(1  +  \frac{{a^*}^2}{{\sigma^*}^2} \frac{1}{(1-\alpha/\beta)^2}\right)^{-1/2}     \left(\tau n\frac{\lambda_\infty}{1-\alpha/\beta}\right)^{1/2}
	\end{aligned},
\end{equation}
and
\begin{equation}
	\begin{aligned}
		G_\pm (\tau,x) = - \left(1  + \frac{{\sigma^*}^2}{{a^*}^2} (1-\alpha/\beta)^2\right)^{-1}d_\pm(\tau,x)  \pm \left(1  +  \frac{{\sigma^*}^2}{{a^*}^2} (1-\alpha/\beta)^2\right)^{-1/2}     \left(\tau n\frac{\lambda_\infty}{(1-\alpha/\beta)^3}\right)^{1/2}
	\end{aligned},
\end{equation}
$d_\pm$ is given in equation (\ref{d_pm}) and $\phi$ is the standard normal probability density function. Let $\alpha/\beta = \hat{\mu}$, we can compute the Greek $C_{\hat{\mu}}$ to measure the parameters in Hawkes process 
\begin{equation}
	\begin{aligned}
		c_{ \hat{\mu}} (T-\tau,x)= x\phi (d_+(\tau,x)) H_+(\tau,x) - Ke^{-r\tau}\phi(d_-(\tau,x))H_-(\tau,x), 
	\end{aligned}
\end{equation}

\begin{equation}
	\begin{aligned}
		H_\pm (\tau,x) =& - \frac{1}{2}   \left[\left(1-\hat{\mu}+\frac{{a^*}^2}{1-\hat{\mu}}\right)^{-1}  + 3\left( \frac{{\sigma^*}^2}{{a^*}^2}(1-\hat{\mu})^3+1-\hat{\mu}\right)^{-1}  \right]d_\pm(\tau,x)  \\& \pm \frac{1}{2} \left(\tau n{\lambda_\infty}\right)^{1/2} \left(\frac{{\sigma^*}^2}{1-\hat{\mu}} +  \frac{3{a^*}^2}{(1-\hat{\mu})^3} \right)\left(  {\sigma^*}^2(1-\hat{\mu})  + \frac{{a^*}^2}{1-\hat{\mu}}  \right)^{-1/2}    
	\end{aligned}.
\end{equation}


\subsection{Implied Volatility and Implied Order Flow}

In this section, we discuss the implied volatility and implied order flow. It reveals the relationship between stock volatility and the order flow in the limit order book system. In this way, the Hawkes-based model can provide more market forecast information than the classical Black-Scholes model. The stock price dynamics modelled by the EMGCHP 

\begin{equation} \label{sde_gbm_2}
\begin{aligned}
dA_{nt} &= A_{nt} \left(  a^*n \frac{\lambda_\infty}{1-\alpha/\beta} + n{\sigma^*}^2 \frac{\lambda_\infty}{2(1-\alpha/\beta)} + n{a^*}^2 \frac{\lambda_{\infty}}{2(1-\alpha/\beta)^3}\right) dt \\& + A_{nt} \sqrt{n{\sigma^*}^2 \frac{\lambda_\infty}{1-\alpha/\beta} + n{a^*}^2 \frac{\lambda_{\infty}}{(1-\alpha/\beta)^3}} dW_t 
\end{aligned}
\end{equation}
is a geometric Brownian motion with volatility
\begin{equation}\label{vol_dy1}
\hat{\sigma} =  \sqrt{{\sigma^*}^2 \frac{\lambda_\infty}{1-\alpha/\beta} + {a^*}^2 \frac{\lambda_{\infty}}{(1-\alpha/\beta)^3}}.
\end{equation}

According to the classical Black-Scholes theory, we can compute the implied volatility $\hat{\sigma}_{implied}$ for the EMGCHP model as well. The idea is the option valuation formula 
\begin{equation}
\begin{aligned}
c(t,x)=x \Phi(d_+(T-t,x)) - Ke^{-r(T-t)}\Phi(d_-(T-t,x)), \,\,\, 0\leq t <T,\,\,\,x>0,
\end{aligned}
\end{equation}
where 
\begin{equation}
d_{\pm}(\tau,x) = \frac{1}{\hat{\sigma}} \left[\log \frac{x}{K} +  \left(r\pm  (\frac{1}{2}\hat{\sigma}^2)\right)\tau  \right],
\end{equation}
provides a relationship between option price $c$, time to expiry $\tau = T-t$ (or time $t$), spot price $x$, strike price $K$, interest rate $r$, and volatility $\hat{\sigma}$, namely 
\[c = f(\tau, S, K, \hat{\sigma},r)=x \Phi(d_+(T-t,x)) - Ke^{-r(T-t)}\Phi(d_-(T-t,x)).\]

Then, given $c$, $S$, $\tau$, $K$, and $r$ we can always solve $c = f(\tau, S, K, \hat{\sigma},r)$ to find the implied volatility $\hat{\sigma}_{implied}$. Since $f(\cdot)$ is an implicit equation with respect to $\hat{\sigma}$, we cannot simply rearrange it to isolate and solve for $\hat{\sigma}$. Some numerical methods would be applied here, such as bisection, fixed-point iteration, and Newton-Raphson method.

If we only focus on the implied volatility $\hat{\sigma}_{implied}$, everything here (implied volatility surface and volatility smile) would be the same as the classical Black-Scholes model. Because all quantities ($c$, $S$, $\tau$, $K$, and $r$) and formulas (the option valuation formula) are same. However, the Hawkes-based model reveals more interesting dynamics about the stock volatility and order flow.

Recall equation (\ref{vol_dy1}), the left-hand side $\hat{\sigma}$ is the stock price volatility. On the right-hand side, $\sigma^*$ and $a^*$ are quantities for the changing size of stock price; $\lambda_{\infty}$, $\alpha/\beta$ are quantities to describe the order flow. If we can compute the implied volatility $\hat{\sigma}_{implied}$ via the option valuation formula and have some estimation of  $\sigma^*$ and $a^*$, we can derive a relationship between order flow and the implied volatility   
\begin{equation}\label{vol_dy2}
\hat{\sigma}_{implied} =  \sqrt{{\sigma^*}^2 \frac{\lambda_\infty}{1-\alpha/\beta} + {a^*}^2 \frac{\lambda_{\infty}}{(1-\alpha/\beta)^3}}.
\end{equation}
Here $\frac{\lambda_\infty}{1-\alpha/\beta}$ is the expectation of the order flow $E(N_t)$ and $\frac{\lambda_{\infty}}{(1-\alpha/\beta)^3}$ is an approximation of the variance of the order flow $Var(N_t)$. So we can rewrite (\ref{vol_dy2}) in the form of 
\begin{equation}\label{vol_dy3}
\hat{\sigma}_{implied}^2 =  {\sigma^*}^2 E(N_t)  + {a^*}^2 Var(N_t).
\end{equation}

Moreover, if we have some estimation about $\alpha/\beta$ (from the historical data, estimation of the $\lambda_{\infty}$ can also work), we can compute the implied expectation and variance of the order flow in the form of
\begin{equation}
E(N_t)_{implied} = \hat{\sigma}_{implied}^2  \left( {\sigma^*}^2 + \frac{{a^*}^2}{(1-\alpha/\beta)^2}\right)^{-1},
\end{equation}
\begin{equation}
Var(N_t)_{implied} = \hat{\sigma}_{implied}^2  \left( {\sigma^*}^2(1-\alpha/\beta)^2 + {a^*}^2\right)^{-1}.
\end{equation}
We call $E(N_t)_{implied}$ and $Var(N_t)_{implied} $ implied order flow expectation and implied order flow variance. Similar as the implied volatility, implied order flow expectation and variance are forward-looking and subjective measures for orders. They reflect the market's forecast of a likely movement for orders in the limit order books. Simulation examples can be found in Figure \ref{implied_orderflow} and Figure \ref{implied_orderflowva}. Since we don't have actual high frequency option data, we applied the option values computed from Figure \ref{option_ST_h} to calculate the implied volatility and set $\alpha = 0.7$, $\beta = 1$, $K=50$, spot price $S = 55$.

\begin{figure}[H]
	\centering
	\includegraphics[width=10cm]{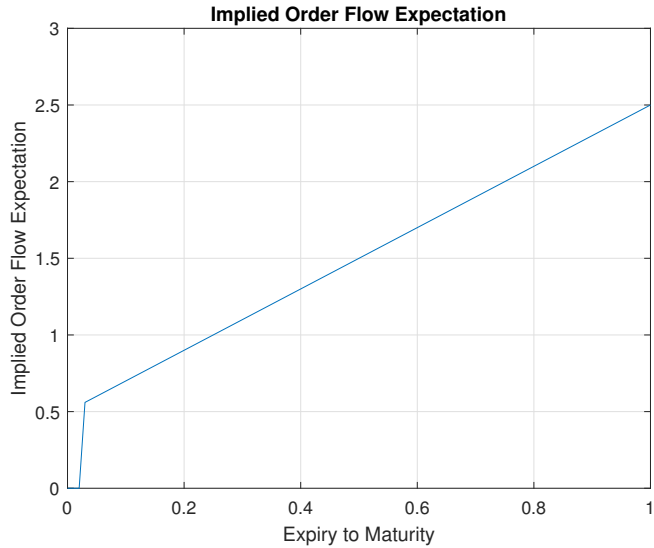}
	\caption{Implied order flow expectation.}
	\label{implied_orderflow}
\end{figure}

\begin{figure}[H]
	\centering
	\includegraphics[width=10cm]{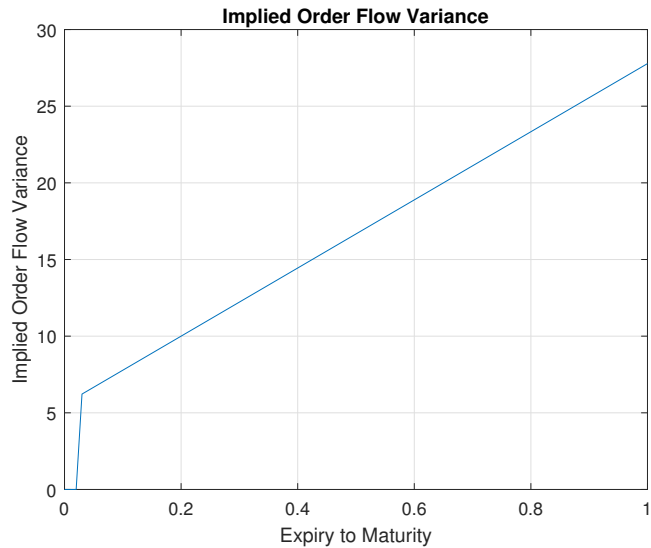}
	\caption{Implied order flow variance.}
	\label{implied_orderflowva}
\end{figure}

\section{Margrabe's Spread Options Valuations with Hawkes-based Models for Two Assets}

There are many kinds of spread options. For instance, in the foreign exchange market 
(spread involves rates in different countries); in the fixed income market (spread between different maturities, e.g., in United States of America, it is Notes-Bonds (NOB) spread); there is also spread between quality levels, e.g., Treasury Bills-EuroDollars (TED spread); in the agricultural futures markets (e.g., soybean complex spread and corn spread); in the energy markets, there are crack spread options (e.g., gasoline crack spread, heating oil crack spread) or spark spread-converting a specific fuel (e.g., natural gas) into electricity. 
We can find this primary cross-commodity transaction in electric energy markets. For more details, see \cite{Carmona}. Note that these spread options are particular cases of exchange option \cite{Rem2013}, Section 2.4.3, where all formulas are given. There is no need to re-derive them.  See also quanto options in 2.4.4 and examples there, e.g., option paying in a foreign currency, etc.  In the same book, section 2.5 deals with Greeks in a multivariate setting

\subsection{Recap: Margrabe's Spread Options Valuations}

Margrabe's approach (\cite{M}) to spread option valuation is to price
\begin{equation}
e^{-rT}{\rm \bf E}^Q[\max(X_T-Y_T,0)],
\label{eq1}
\end{equation}
where $X_t$ and $Y_t$ are two assets for spot prices under the risk-neutral measure, and they are modelled by two geometric Brownian motions (GBMs):
\begin{equation}
\begin{array}{rcl}
dX_t&=&rX_tdt+\sigma_XX_tdB_t^X\\
dY_t&=&rY_tdt+\sigma_YY_tdB_t^Y.\\
\end{array}
\label{eq2}
\end{equation}

We denote by $r$ the interest rate and $T$ a maturity, and the correlation between the two Brownian motions $B_t^X$ and $B_t^Y$ is a constant $\rho.$
The change of numeraire is the most efficient method to solve this problem. Let us take $Y_t$ as the numeraire. Then, in the $Y$ measure, ${\rm \bf E}^Q_Y,$ all assets discounted by $Y$ must be martingales. Let $V(t,X_t,Y_t)$ be the value of the option at time $t.$ Since $V$ must be an ${\rm \bf E}^Q_Y$ martingale, we have:
$$
\frac{V(0,X_0,Y_0)}{Y_0}={\rm \bf E}^Q_Y\left[\frac{V(T,X_T,Y_T)}{Y_T}\right],
$$
from which we derive
\begin{equation}
V(0,X_0,Y_0)=Y_0 {\rm \bf E}^Q_Y\left[\max\left(\frac{X_T}{Y_T}-1,0\right)\right].
\label{eq3}
\end{equation}

We note that
\begin{equation}
\frac{X_T}{Y_T}=\frac{X_0}{Y_0}e^{\sigma_XB_T^X-\sigma_YB_T^Y+({\rm It\hat{o}\, drift\, term})}.
\label{eq4}
\end{equation}

We also note that the variance $\sigma^2_{exp}$ for the exponent in (\ref{eq4}) is:
\begin{equation}
\begin{array}{rcl}
\sigma^2_{exp}&=&Var[\sigma_XB_T^X-\sigma_YB_T^Y+({\rm It\hat{o}\, drift\, term})]\\
&=&Var[\sigma_XB_T^X-\sigma_YB_T^Y]\\
&=&T(\sigma^2_X+\sigma^2_Y-2\rho\sigma_X\sigma_Y).\\
\end{array}
\label{eq5}
\end{equation}

Since $ \frac{X_T}{Y_T}$ is log-normal distributed with variance $\sigma^2_{exp},$ we have:
$$
{\rm \bf E}\left[\max\left(\frac{X_T}{Y_T}-K,0\right)\right]={\rm \bf E}\left[\frac{X_T}{Y_T}\right]N(d_1)-KN(d_2),
$$
where $K$ is a strike price, and
\begin{equation}
d_{1,2}=\frac{\log[E(X_T/Y_T)/K\pm\frac{1}{2}\sigma^2_{exp}]}{\sigma_{exp}}.
\label{eq6}
\end{equation}

For notational convenience, we denote $d_{1,2}=d_1$ for the ``plus'' sign, and $d_{1,2}=d_2$ for the ``minus'' sign.

Collecting Equations~(\ref{eq3})--(\ref{eq6}), we finally have:
\begin{equation}
V(0,X_0,Y_0)=Y_0\left[{\rm \bf E}^Q_Y\left(\frac{X_T}{Y_T}N(d_1)-N(d_2)\right)\right]=X_0N(d_1)-Y_0N(d_2),
\label{eq7}
\end{equation}
where
\begin{equation}
d_{1,2}=\frac{\log(X_0/Y_0)\pm\frac{1}{2}\sigma^2_{exp}}{\sigma_{exp}}.
\label{eq8}
\end{equation}

We observe that there is no discounting because funding is embedded in the asset $Y.$

Remark. If we compare our result 
 with the Black--Scholes formula, the value of the option of a call in \mbox{Equation~(\ref{eq7})} is similar to BS as if we used BS with $X$ as the underlying asset and $Y_0$ as the strike price, with volatility $\sigma_{exp}.$ We note that $\sigma^2_{exp}$ is the variance of $\log(X_T/Y_T).$

\subsection{Dynamics for Two Hawkes-based Assets}

 We suppose that risk-free asset $B_t$ and two risky assets $S_i(t)$ follow the following dynamics, respectively (see \cite{S}):
$$
\left\{
\begin{array}{rcl}
B(t)&=&B_0\exp\{rt\}\\
S_i(t)&=&S_i(0)\exp\{G_i(t)\},\\
\end{array}
\right.
$$
where $G_i(t):=\sum_{k=1}^{N_i(t)}a_i(X^i_k)$ are two GCHP, $i=1,2,$ $N_i(t)$ are two independent Hawkes processes with background intensities $\lambda_i$ and self-exciting functions $\mu_i(t),$ $X^i_k$ are two independent discrete-time finite or infinite state Markov chains with state space $X=\{1,2,3,...,N\}$ or $X=\{1,2,...,N,...\},$ respectively, $r>0$ is the interest rate.

$G_i(t)$ can be approximated as (see )
 $$
 G_i(t)\approx a_i^*\frac{\lambda_i}{1-\hat\mu_i}t+\bar\sigma_i W_i(t),
  $$
where $\hat\mu_i:=\int_0^{+\infty}\mu_i(s)ds<1,$ $a_i^*$ is an average of $a(_ix)$ over stationary distributions of MCs $X^i_k,$ $\lambda_i$ is a background intensity, $\bar\sigma_i>0$ is defined as follows 

$$
\bar\sigma_i=\sqrt{(\sigma_i^*)^2+\Big(a_i^*\sqrt{\frac{\lambda_i}{(1-\hat\mu_i)^3}}\Big)^2},
$$
where $\sigma_i^*=\sigma_i\sqrt{\lambda_i/(1-\hat\mu_i)}.$  We note, that for exponential decaying intensity $\hat\mu_i=\alpha_i/\beta_i.$

Thus, by the Theorem 2 (see \cite{S})  $S_i(t)$ can be approximated by the pure diffusion process:
$$
S_i(t)\approx S_i(0)+a^*\frac{\lambda_i}{1-\hat\mu_i}t+\bar\sigma_i W_i(t),
$$
where $W_i(t)$ are two standard Wiener process with correlation coefficient $\rho.$ 

Therefore, $S_i(t)$ can be presented in the following way:
$$
S_i(t)=S_i(0)e^{a_i^*\frac{\lambda_i}{1-\hat\mu_i}t+\bar\sigma_i W_i(t)}.
$$
Using It$\hat o$ formula we can get:
$$
dS(_it)=S_i(t)[(a_i^*\frac{\lambda_i}{1-\hat\mu_i}+\frac{\bar\sigma_i^2}{2})dt+\bar\sigma_i dW_i(t)].
$$
We note, that in risk-neutral world the dynamics for two Hawkes-based assets is
$$
dS(_it)=S_i(t)[rdt+\bar\sigma_i d\hat W_i(t)],
$$
where $\hat W_i(t)$ are two Girsanov's Wiener processes.

\subsection{Margrabe's Spread Options Valuations with Hawkes-based Models for Two Assets}

According to Section 6.2, Margrabe's spread option valuation formula for two Hawkes-based assets in $S_i(t)$ has the following form
$$
V(0,S_1(0),S_2(0))=S_1(0)N(d_1)-S_2(0)N(d_2),
$$ 
where
$$
d_{1,2}=\frac{\ln(S_1(0)/S_2(0))\pm\frac{1}{2}\bar\sigma^2_{exp}}{\bar\sigma_{exp}},
$$
and
$$
\bar\sigma_{exp}^2=Var(\bar\sigma_1\hat W_i(T)-\bar\sigma_2\hat W_2(T))=T[\bar\sigma_1^2+\bar\sigma_2^2-2\rho\bar\sigma_1\bar\sigma_2].
$$
Here, $N(x)$ is the normal distribution function.

\subsection{Numerical Example}

Below we present a figure for spread option of two $1D$ EMGCHP. 

For $1D$ case parameters are:

alpha1 = 0.7;
beta1 = 1;
alpha2 = 0.5 ;
beta2 = 1;
lambda1 = 0.75;
sig1 = 0.05;
a1 = 0.03;
lambda2 = 1.5;
sig2 = 0.05;
a2 = 0.01;
S1 = 30;
S2 = 20;
T = 1:1:100;
rho = 0:0.01:1;

\begin{figure}[H]
	\centering
	\includegraphics[width=10cm]{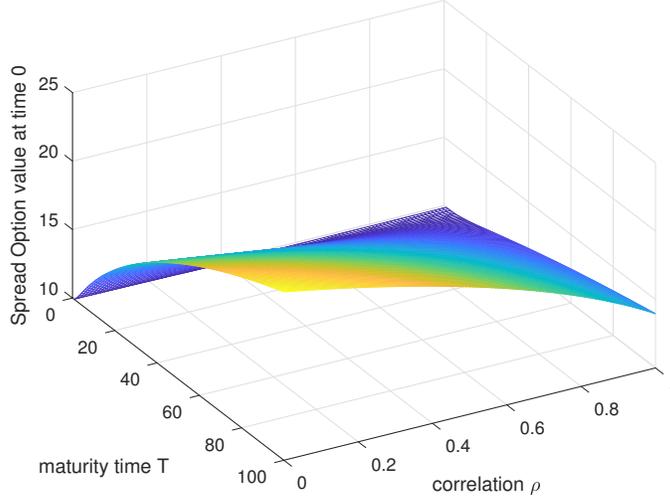}
	\caption{ Two $1D$ EMGCHP spread option.}
	\label{option_Spread_1d}
\end{figure}


\section{Spread Option Pricing with the Two-dimensional \\EMGCHP}

We consider the spread option pricing with a two-dimensional EMGCHP. Let $\vec{A}_t = (A_{1,t}, A_{2,t})$ be a two-dimensional EMGCHP. $A_{1,t}$ and $A_{2,t}$ denote two stock price. Let $\vec{N}_t = (N_{1,t}, N_{2,t})$ be a two-dimensional Hawkes process in the EMGCHP. Recall the FCLT II approximation, we have 
\begin{equation}
\vec{A}_{nt}  \sim \exp \left[ \sqrt{n} \Tilde{\sigma}^* \boldsymbol{\Sigma^{1/2}} \Vec{W}_1(t) + \sqrt{n} \Tilde{a}^* (\mathbf{I}-\mathbf{K})^{-1} \boldsymbol{\Sigma^{1/2}} \Vec{W}_2(t) +   \Tilde{a}^* n  (\mathbf{I}-\mathbf{K})^{-1}\vec{\lambda}_\infty)t \right], 
\end{equation} 
for all $t>0$ and some large enough $n$. Here, $\Vec{W}_1(t)$ and $\Vec{W}_2(t)$ are two independent standard $2$-dimensional Brownian motions. $\sqrt{n} \Tilde{\sigma}^* \boldsymbol{\Sigma^{1/2}}$ and $\sqrt{n} \Tilde{a}^* (\mathbf{I}-\mathbf{K})^{-1} \boldsymbol{\Sigma^{1/2}}$ are $2\times 2$ matrices.

Let $C$ be a concatenation of matrices $\sqrt{n} \Tilde{\sigma}^* \boldsymbol{\Sigma^{1/2}}$ and $\sqrt{n} \Tilde{a}^* (\mathbf{I}-\mathbf{K})^{-1} \boldsymbol{\Sigma^{1/2}}$ in the form of  
\[C = \left[ \sqrt{n} \Tilde{\sigma}^* \boldsymbol{\Sigma^{1/2}},  \sqrt{n} \Tilde{a}^* (\mathbf{I}-\mathbf{K})^{-1} \boldsymbol{\Sigma^{1/2}} \right] \in \mathbb{R}^{2 \times 4} \] 
and $\vec{W}_t$ be a $4$-dimensional standard Brownian motion, then 
\[\vec{A}_t = \exp \left[ { C\vec{W}_t + \Tilde{a}^* n  ((\mathbf{I}-\mathbf{K})^{-1}\vec{\lambda}_\infty)t} \right] .\]
By applying the multivariate It\^{o} formula, we derive the price process in the form a multivariate geometric Brownian motion  
\begin{equation}
dA_{i,t} = A_{i,t} \left\{(  a^*_in  ((\mathbf{I}-\mathbf{K})^{-1}\vec{\lambda}_\infty)_i + \frac{1}{2}(CC')_{ii} ) dt +  C_i d\vec{W}_t \right\},
\end{equation} 
where $C_i$ is the $i$th row of $C$.    

To simplify notations, we let the drift for $i$th asset $ a^*_in  ((\mathbf{I}-\mathbf{K})^{-1}\vec{\lambda}_\infty)_i + \frac{1}{2}(CC')_{ii} = D_i$ and $\vec{W}_t = ({W}^1_t, {W}^2_t, {W}^3_t,{W}^4_t)'$. Then, price processes for two stocks are in the form of

\begin{equation}\label{asstes_1}
	\begin{cases}
		dA_{1,t} = A_{1,t} \left(D_1dt +  C_{11}d{W}^1_t + C_{12}d{W}^2_t +C_{13}d{W}^3_t+C_{14}d{W}^4_t\right)\\
			dA_{2,t} = A_{2,t} \left(D_2dt +  C_{21}d{W}^1_t + C_{22}d{W}^2_t +C_{23}d{W}^3_t+C_{24}d{W}^4_t\right)
	\end{cases},
\end{equation}
where $C_{ij}$ is the $i,j$th element of matrix $C$. Let  
\[
\tilde{C}_1 = \sqrt{C_{11}^2+C_{12}^2+C_{13}^2+C_{14}^2} \,\,\,,
\tilde{C}_2 = \sqrt{C_{21}^2+C_{22}^2+C_{23}^2+C_{24}^2}
\]
then
\[
  dB^1_t = \frac{1}{\tilde{C}_1}  \left(C_{11}d{W}^1_t + C_{12}d{W}^2_t +C_{13}d{W}^3_t+C_{14}d{W}^4_t\right)  
\]
and 
\[
dB^2_t = \frac{1}{\tilde{C}_2}  \left(C_{21}d{W}^1_t + C_{22}d{W}^2_t +C_{23}d{W}^3_t+C_{24}d{W}^4_t\right)  
\]
are standard Brownian motions with correlation
\begin{equation}\label{correlation1}
	\rho = dB^1_tdB^2_t = \frac{1}{\tilde{C}_2\tilde{C}_1}\left(C_{11}C_{21}+ C_{12}C_{22}+C_{131}C_{23}+C_{14}C_{24}\right).
\end{equation}
\begin{rem}
	Different with two $1$-dimensional models, the correlation $\rho$ for the two-dimensional EMGCHP model were given by the multivariate Hawkes process, namely $\rho$ is defined by matrix $C$ and $C$ can be computed by parameters of the Hawkes process and Markov chains.  
\end{rem}
Then, we can rewrite equation (\ref{asstes_1}) into the form of
\begin{equation}\label{asstes_2}
\begin{cases}
dA_{1,t} = A_{1,t} \left(D_1dt +  \tilde{C}_1 dB_t^1\right)\\
dA_{2,t} = A_{2,t} \left(D_2dt +  \tilde{C}_2 dB_t^2\right)
\end{cases},
\end{equation}
where $D_i$, $\tilde{C}_1$, and correlation between $B_t^1$ and $B_t^2$ are defined before. Then, we can apply the Margrabe' s formula for spread option pricing. Let $T$ denotes the maturity time, $A_1(0)$ and $A_2(0)$ are initial prices for two stocks, the spread option value at time $0$ is in the form of
\[
V(A_1(0), A_2(0), T, 0) = A_1(0)N(d_+) - A_2(0)N(d_-),
\]
where
$N$ denotes the cumulative distribution function for a standard normal distribution. $d_{\pm}$ is given by
\[
d_{\pm} = \frac{\ln(\frac{A_1(0)}{A_2(0)}) + \frac{T\sigma_{exp}^2}{2}}{\sigma_{exp}\sqrt{T}}
\]
and
\[
\sigma_{exp} = \sqrt{\tilde{C}_1^2+\tilde{C}_2^2 - 2\rho \tilde{C}_1\tilde{C}_2},
\]
where $\rho$ is given by (\ref{correlation1}).


\subsection{Numerical Example}

Below we present a figures for  spread option of one $2D$ EMGCHP.  

For $2D$ case parameters are:

alpha = [115.7317,0.4492;0.0218,123.2703];
beta =  [280.9249,2.9611;0.0669,307.2993];
lambda = [0.0545, 0.0593]';
a = diag([0.03,0.04]);
sig = diag([0.05,0.03]);
S1 = 30;
S2 = 20;
T = 1:1:100;

\begin{figure}[H]
	\centering
	\includegraphics[width=10cm]{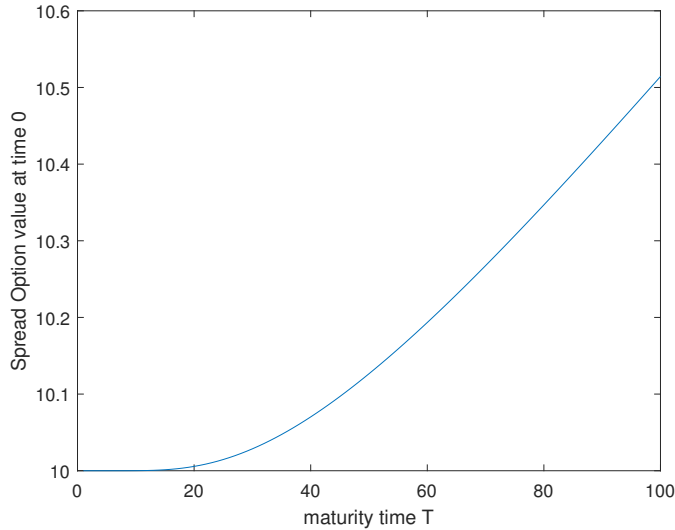}
	\caption{ One $2D$ EMGCHP spread option.}
	\label{option_ST_spread_2d}
\end{figure}

\section{Basket Options Pricing}
In this section, we consider basket options pricing for MDCHP.
(See \cite{KKM}, \cite{Gen1993}).

\subsection{Basket Option Valuations with Multivariate Hawkes-based Models}

We consider a $d$-dimensional EMGCHP to denote $d$ stocks prices. With the FCLT approximation, price dynamics can be written as
\begin{equation}
	dA_{i,t} = A_{i,t} \left\{(  a^*_in  ((\mathbf{I}-\mathbf{K})^{-1}\vec{\lambda}_\infty)_i + \frac{1}{2}(CC')_{ii} ) dt +  C_i d\vec{W}_t \right\},
\end{equation} 
where $C_i$ is the $i$th row of $C$ defined before and $\vec{W}_t$ is a $2d$-dimensional Brownian motion. Similar to the Section 6, we can simplify notations by letting the drift for $i$th asset $ a^*_in  ((\mathbf{I}-\mathbf{K})^{-1}\vec{\lambda}_\infty)_i + \frac{1}{2}(CC')_{ii} = D_i$ and $\vec{W}_t = ({W}^1_t, {W}^2_t, \cdots,{W}^{2d}_t)'$. Then, the $i$th price process is in the form of

\begin{equation}\label{dasstes_1}
		dA_{i,t} = A_{i,t} \left(D_idt +\sum_{j=1}^{2d}C_{ij}d{W}^j_t  \right)
\end{equation}
where $C_{ij}$ is the $i,j$th element of matrix $C$. Let  
\[
\tilde{C}_i =  \left( \sum_{j=1}^{2d}C_{ij}^2 \right)^{1/2}
\]
then
\[
dB^i_t = \frac{1}{\tilde{C}_i}  \left(\sum_{j=1}^{2d}C_{ij}d{W}^j_t \right)  
\]
are standard Brownian motions. And for any $i\neq k$, the correlation between $B^i_t$ and $B^k_t$ is in the form of
\begin{equation}\label{correlation}
	\rho_{ik} = dB^i_tdB^k_t = \frac{1}{\tilde{C}_i\tilde{C}_k}\left(  \sum_{j=1}^{2d}C_{ij}C_{kj} \right).
\end{equation}
Then, we can rewrite equation (\ref{dasstes_1}) into the form of
\begin{equation}\label{dasstes_2}
		dA_{i,t} = A_{i,t} \left(D_idt +  \tilde{C}_i dB_t^i\right).
\end{equation}
By applying the Girsanov theorem, we can rewrite (\ref{dasstes_2}) as
\begin{equation}\label{dasstes_3}
	d \tilde{A}_{i,t} = A_{i,t}rdt +   A_{i,t} \tilde{C}_i d{\tilde{B}}_t^i,
\end{equation}
where ${\tilde{B}}_t^i$ is a Brownian motion under the risk-neutral measure and $r$ is the interest rate.

Next, we can consider the basket option valuation under the price dynamic (\ref{dasstes_3}). Define the price of a basket of stocks by
\[
A_{basket}(T) = \sum_{i=1}^{d} \omega_{i} \tilde{A}_{i,T}.
\]
where $\omega_i$ are weights for each stock and $\sum\omega_i=1$. Then the payoff of the basket option at maturity time $T$ is in the form of
\begin{equation}\label{basketopt}
	P_{\text {Basket }}(A_{basket}(T), K, \theta)= \left[\theta(  A_{basket}(T)-K)\right]^{+},
\end{equation}
where $K$ is the strike price. When $\theta=1$, (\ref{basketopt}) is a call option and when $\theta=-1$, (\ref{basketopt}) denotes a put option.

Rewriting the payoff as

\begin{equation}\label{basketopt2}
	\begin{aligned}
			P_{\text {Basket }}(\sum_{i=1}^{d}\omega_i\tilde{A}_{i,T}, K, \theta) &= \left[\theta(  \sum_{i=1}^{d}\omega_i\tilde{A}_{i,T}-K)\right]^{+} \\& = \left[\theta\left(\left(\sum_{i=1}^{d} w_{i} \tilde{A}_{i,0}e^{rT}\right) \sum_{i=1}^{d} \tilde{\omega}_{i} \tilde{A}_{i,T}^{*}-K\right)\right]^{+}, 
	\end{aligned}
\end{equation}
where 
\[
\tilde{\omega}_{i} = \frac{\omega_{i} \tilde{A}_{i,0}e^{rT}}{\sum_{i=1}^{d} \omega_{i} \tilde{A}_{i,0}e^{rT}} \,\,\,\,\,\,\mathrm{and} \,\,\,\,\,\,\tilde{A}_{i}^{*}(T) = \frac{\tilde{A}_{i,T}^{*}}{\tilde{A}_{i,0}e^{rT}}.
\]
Then, we apply geometric average approximation method proposed by \cite{Gen1993}, namely approximate $\left(\sum_{i=1}^{d} \omega_{i} \tilde{A}_{i,0}e^{rT}\right)\sum_{i=1}^{d} \tilde{\omega}_{i} \tilde{A}_{i,T}^{*}$ by
\begin{equation}
	A_{geometric}(T)=\left(\sum_{i=1}^{d} \omega_{i} \tilde{A}_{i,0}e^{rT}\right) \prod_{i=1}^{d}\left(\tilde{A}_{i,T}^{*}\right)^{\tilde{\omega}_{i}}.
\end{equation}

We note that $A_{geometric}(T)$ follows the log-normally distribution. In this way, the basket option can be valuated by the classical Black-Scholes method 

\begin{equation}
	V_{\text {Basket }}(T)=e^{-rT}	 \theta\left(e^{\tilde{m}+\frac{1}{2} \tilde{v}^{2}} \Phi\left(\theta d_{1}\right)-K^{*} \Phi\left(\theta d_{2}\right)\right),
\end{equation}
where
\begin{equation}
	\begin{aligned}
		d_{1} &=\frac{\tilde{m}-\log K^{*}+\tilde{v}^{2}}{\tilde{v}}, \\
		d_{2} &=d_{1}-\tilde{v}, \\
		\tilde{m} &=\log \left(\sum_{i=1}^{d} \omega_{i} \tilde{A}_{i,0}e^{rT}\right)-\frac{1}{2} \sum_{i=1}^{d} \tilde{\omega}_{i} \tilde{C}_{i}^{2} T  \\
		\tilde{v}^{2} &=\sum_{i=1}^{d} \sum_{k=1}^{d} \tilde{\omega}_{i} \tilde{\omega}_{k} \tilde{C}_{i} \tilde{C}_{k} \rho_{i k} T\\ K^{*} &= K-(\mathrm{E}(A_{basket}(T))-\mathrm{E}(A_{geometric}(T)))
	\end{aligned}
\end{equation}
$\Phi$ is the cumulative distribution of standard normal distribution.


\subsection{Numerical Example: Three Stocks Basket Option Valuation}

Parameters under consideration:
($\alpha_{ij}/\beta_{ij}$) =\\
$[  0.5933\quad  0.2068\quad  0.1743;\quad
0.0845\quad  0.6746\quad 0.1222;\quad
0.0312\quad 0.2033\quad 0.3820 ];$\\
 
$\lambda = [0.0545,\quad 0.0593,\quad 0.0623]';$        (intensity background)\\
Markov Chain parameters:\\
$a = diag([0.03,0.04,0.02]);\quad \Sigma = diag([0.05,0.03,0.06]);$\\
$A_0 = [30, 20, 25];\quad  (stock prices at time 0)$\\
$r = 0.06;$\quad                   (interest rate)\\
$\Omega = [0.3, 0.5, 0.2];$\quad  (weight for basket stocks).
 \begin{figure}[H]
	\centering
	\includegraphics[width=10cm]{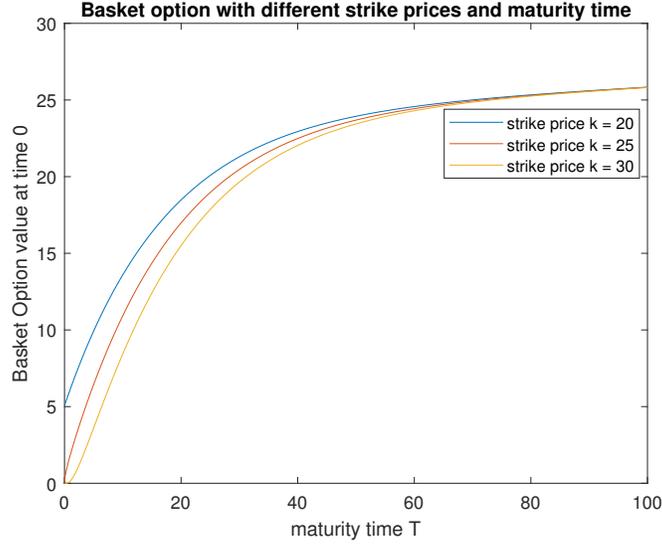}
	\caption{Basket Option Value}
	\label{option_bask1}
\end{figure}


\section{Hedging with EMGCHP}

In this section, we calculate minimal hedging portfolio and capital for the EMGCHP model. 
(See \cite{HP}, \cite{HK}, \cite{BS}).

The minimal hedge portfolio is $\pi_t=(\alpha_t,\beta_t),$ where
$$
\alpha_t=N \left(    \frac{\ln(S_t/K)+(T-t)(r+\sigma^2/2)}{\sigma\sqrt{T-t}}\right)
$$
and 
$$
\beta_t=-\frac{K}{e^{rT}}N\left(\frac{\ln(S_t/K)+(T-t)(r-\sigma^2/2)}{\sigma\sqrt{T-t}}\right).
$$
The capital is:
$$
X_t=S_tN\left(\frac{\ln(S_t/K)+(T-t)(r+\sigma^2/2)}{\sigma\sqrt{T-t}})-Ke^{-r(T-t)}N(\frac{\ln(S_t/K)+(T-t)(r-\sigma^2/2)}{\sigma\sqrt{T-t}}\right)
$$

Below are some graphs for the dynamics of $\alpha_,\beta_t, X_t$ and $S_t.$

\begin{figure}[H]
	\centering
	\includegraphics[width=10cm]{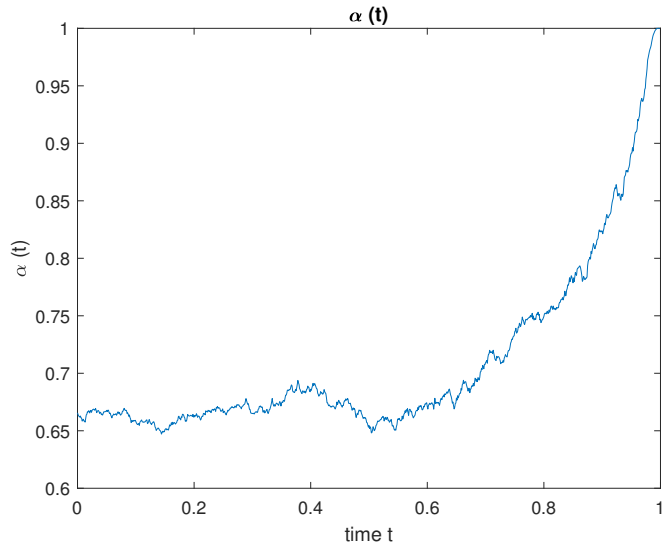}
	\caption{ Dynamics for $\alpha_t$}
	\label{option_h1}
\end{figure}

\begin{figure}[H]
	\centering
	\includegraphics[width=10cm]{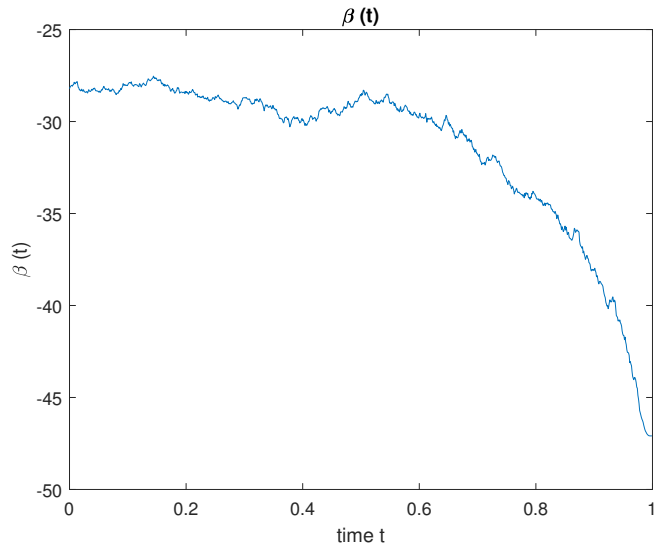}
	\caption{ Dynamics for $\beta_t$}
	\label{option_h2}
\end{figure}

\begin{figure}[H]
	\centering
	\includegraphics[width=10cm]{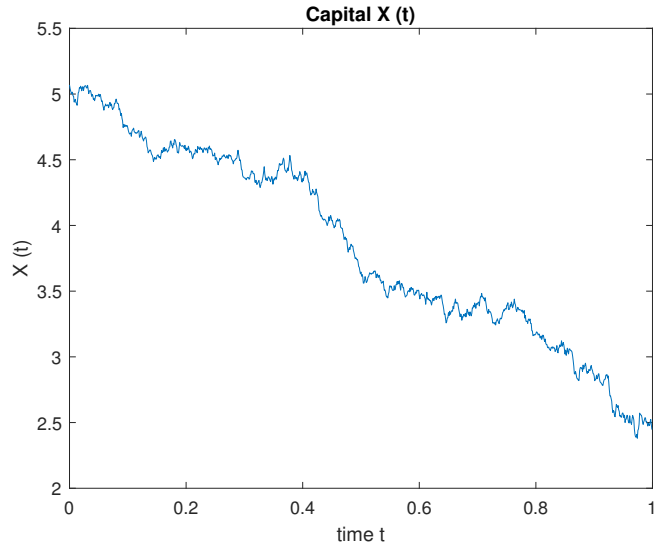}
	\caption{ Dynamics for capital $X_t$}
	\label{option_h3}
\end{figure}

\begin{figure}[H]
	\centering
	\includegraphics[width=10cm]{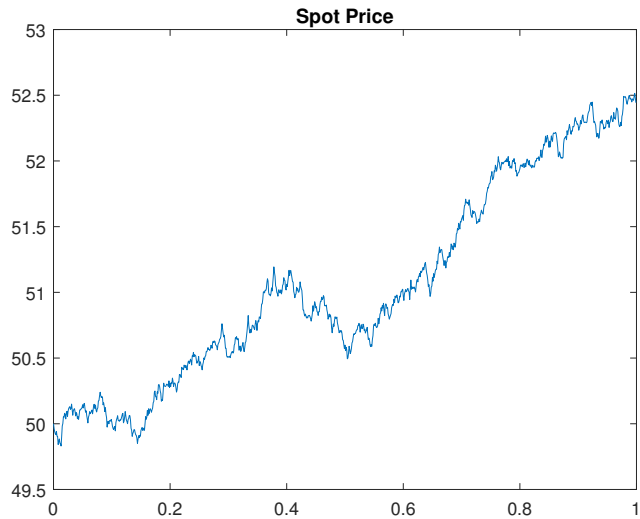}
	\caption{ Evolution of the spot price}
	\label{option_h4}
\end{figure}


\section{Conclusion} 

In this paper, we considered pricing of European options and spread options for Hawkes-based model for the limit order book. We introduced multivariate Hawkes process and the multivariable general compound Hawkes process. Exponential multivariate general compound Hawkes processes and limit theorems for them, namely, LLN and FCLT, have been considered then. We also considered a special case of one-dimensional EMGCHP and its limit theorems. Option pricing with $1D$ EGCHP in LOB, hedging strategies, and numerical example have been presented. We also introduced greeks calculations for those models. Margrabe's spread options valuations with Hawkes-based models for two assets and numerical example were presented. Also, Margrabe's spread option pricing with two $2D$ EMGCHP and numerical example have been included. Basket options valuations with numerical example are included. We finally discussed the implied volatility and implied order flow. It reveals the relationship between stock volatility and the order flow in the limit order book system. In this way, the Hawkes-based model can provide more market forecast information than the classical Black-Scholes model.

\section{Acknowledgements:} 
All authors wish to thank NSERC for continuing support.


\end{document}